\documentclass[11pt]{article}
\usepackage{theorem}
\usepackage{subfigure}
\usepackage{epsfig}
\usepackage{amssymb}
\usepackage{fullpage}
\usepackage{amsmath}
\usepackage{ifthen}
\usepackage{verbatim}
\usepackage{mdframed}
\usepackage{xspace}
\usepackage{enumitem}
\usepackage[hyphens]{url}
\usepackage{bm}

\usepackage[hyphens]{url}
\usepackage{hyperref}
\usepackage[hyphenbreaks]{breakurl}

\usepackage[nottoc]{tocbibind}

\usepackage{tikz}
\usetikzlibrary{shapes,shapes.geometric,arrows,fit,calc,positioning,automata,arrows.meta}
\tikzset{>={Latex[width=2mm,length=2mm]}}
\usepackage{float}

\usepackage[noend,noline]{algorithm2e}
\SetEndCharOfAlgoLine{}
\SetArgSty{}
\SetKwBlock{Repeat}{repeat}{}
\newtheorem{theorem}	 			{Theorem}[section]

\newtheorem{corollary}		[theorem]	{Corollary}
\newtheorem{prop}		[theorem]	{Proposition}
{\theorembodyfont{\rmfamily} \newtheorem{definition}
[theorem]	{Definition}}
{\theorembodyfont{\rmfamily} \newtheorem{remark}		[theorem]
{Remark}}
{\theorembodyfont{\rmfamily} \newtheorem{example}		[theorem]
{Example}}
{\theorembodyfont{\rmfamily} }
{\theorembodyfont{\rmfamily} }
{\theorembodyfont{\rmfamily} \newtheorem{q}			[theorem]
{Question}}
{\theorembodyfont{\rmfamily} }
{\theorembodyfont{\rmfamily} }
\theoremstyle{break}
{\theorembodyfont{\rmfamily} }

\newenvironment{proof}{\noindent {\em {Proof:}}}{$\blacksquare$\vskip
\belowdisplayskip}

\usepackage{tablefootnote} 
\makeatletter 
\AfterEndEnvironment{mdframed}{%
 \tfn@tablefootnoteprintout%
 \gdef\tfn@fnt{0}%
}
\makeatother

\newlist{exlist}{enumerate}{1}
\setlist[exlist]{label=(\alph*)}

\newcommand{\prob}[2][]{\text{\bf Pr}\ifthenelse{\not\equal{}{#1}}{_{#1}}{}\!\left[#2\right]}
\newcommand{\expect}[2][]{\text{\bf E}\ifthenelse{\not\equal{}{#1}}{_{#1}}{}\!\left[#2\right]}
\newcommand{\var}[2][]{\text{Var}\ifthenelse{\not\equal{}{#1}}{_{#1}}{}\!\left[#2\right]}
\newcommand{\dev}[2][]{\text{StdDev}\ifthenelse{\not\equal{}{#1}}{_{#1}}{}\!\left[#2\right]}

\mdfsetup{frametitlealignment=\centering}
\mdfdefinestyle{offset}{backgroundcolor=white,linecolor=black,innerrightmargin=15pt,leftmargin=23pt,rightmargin=23pt,innertopmargin=.5\baselineskip,innerbottommargin=.5\baselineskip}

\def\sse{\subseteq}

\def\var{\mbox{Var}}

\newcommand{\tfm}{(\allocs,\prices,\burns)}
\newcommand{\fpa}{(\allocs^f,\prices^f,\burns^f)}

\newcommand\fnsep{\textsuperscript{,}}

\newcommand{\bid}{b}
\newcommand{\bids}{{\mathbf \bid}}

\newcommand{\bidt}[1][t]{{\bid_{#1}}}

\newcommand{\epnebid}{b^*}

\newcommand{\val}{v}

\newcommand{\valt}[1][t]{{\val_{#1}}}

\newcommand{\alloc}{x}
\newcommand{\allocs}{{\mathbf \alloc}}

\newcommand{\alloct}[1][t]{\alloc_{#1}}

\newcommand{\price}{p}
\newcommand{\prices}{{\mathbf \price}}
\newcommand{\pricet}[1][t]{\price_{#1}}

\newcommand{\burn}{q}
\newcommand{\burns}{{\mathbf \burn}}
\newcommand{\burnt}[1][t]{\burn_{#1}}

\newcommand{\cost}{\gamma}
\newcommand{\oca}{(\bids,\bm{\tau})}

\newcommand{\nine}{*}
\newcommand{\ninetfm}{(\allocs^{\nine},\prices^{\nine},\burns^{\nine})}

\title{Transaction Fee Mechanism Design for the Ethereum Blockchain:
  An Economic Analysis of EIP-1559\thanks{This work was 
funded by the Decentralization Foundation
(\protect\url{https://d24n.org/}).  
This report has benefited from comments by and discussions with a
number of people: Maryam Bahrani, Abdelhamid Bakhta,
Tim Beiko, Vitalik Buterin, Matheus Ferreira, Danno Ferrin,
James Fickel, Hasu, Georgios
Konstantopoulos, Andrew Lewis-Pye,
Barnab\'e Monnot, Daniel Moroz, Mitchell Stern, Alex Tabarrok,
and Peter Zeitz.
Thanks to James also for introducing me to the problem.}}
\author{Tim
Roughgarden\thanks{Author's permanent position: Professor of Computer Science,
Columbia University, 500 West 120th Street, New York, NY 10027.
Email: {\tt tim.roughgarden@gmail.com}.  Disclosure: I have
no financial interests in Ethereum, either long or short.}}
\date{\today}

\begin{document}

\maketitle

\begin{abstract}
EIP-1559 is a proposal to make several tightly coupled additions to
Ethereum's transaction fee mechanism, including variable-size blocks
and a burned base fee that rises and falls with demand.  This report
assesses the game-theoretic strengths and weaknesses of the proposal
and explores some alternative designs.
\end{abstract}

\setcounter{tocdepth}{2}
\tableofcontents

\section{TL;DR}

\subsection{A Brief Description of EIP-1559}

In the Ethereum protocol, the transaction fee mechanism is the
component that determines, for every transaction added to the Ethereum
blockchain, the price paid by its creator.  Since its inception,
Ethereum's transaction fee mechanism has been a {\em first-price
  auction}: Each transaction comes equipped with a bid,
corresponding to the gas limit times the gas price, which is 
transferred from its creator to the miner of the block that includes
it.

EIP-1559 proposes a major change to Ethereum's transaction fee
mechanism.  Central to the design is a {\em base fee}, which plays the
role of a reserve price and is meant to match supply and demand.  Every
transaction included in a block must pay that block's base fee (per
unit of gas), and this payment is burnt rather than transferred to the
block's miner.  Blocks are allowed to grow as large as double a target
block size; for example, with a target of 12.5M gas, the maximum block
size would be 25M gas.  The base fee is adjusted after every block,
with larger-than-target blocks increasing it and smaller-than-target
blocks decreasing it.  Users seeking special treatment, such as
immediate inclusion in a period of rapidly increasing demand or a
specific position within a block, can supplement the base fee with a
transaction tip that is transferred directly to the miner of the block
that it includes it.

\subsection{Ten Key Takeaways}

The following list serves as an executive summary for 
busy readers as well as a road map for those wanting to dig deeper.
\begin{enumerate}

\item No transaction fee mechanism, EIP-1559 or otherwise, is likely
  to substantially decrease average transaction fees; persistently
  high transaction fees is a scalability problem, not a mechanism
  design problem.  (See Section~\ref{sss:toohigh} for details.)

\item EIP-1559 should decrease the variance in transaction fees
and the delays experienced by some users  through the flexibility of  variable-size blocks.  (Section~\ref{sss:benefits})

\item EIP-1559 should improve the user experience through easy fee
  estimation, in the form of an ``obvious optimal bid,'' outside of
  periods of rapidly increasing demand.  (Section~\ref{ss:1559uic})

\item The short-term incentives for miners to carry out the protocol
  as intended are as strong under EIP-1559 as with first-price
  auctions.  (Sections~\ref{ss:1559mmic} and~\ref{ss:1559oca})

\item The game-theoretic impediments to double-spend attacks,
  censorship attacks,
denial-of-service attacks, and
long-term revenue-maximizing strategies such as base fee
manipulation appear as strong under EIP-1559 as with first-price
auctions.  (Section~\ref{ss:1559collusion})

\item EIP-1559 should at least modestly decrease the rate of ETH inflation
 through the burning of transaction fees.  (Section~\ref{ss:benefits})

\item The seemingly orthogonal goals of easy fee estimation and fee
  burning are inextricably linked through the threat of
  off-chain agreements. (Sections~\ref{ss:refund}--\ref{ss:fpa_burn})

\item Alternative designs include paying base fee revenues forward to
  miners of future blocks rather than burning them; and replacing
  variable user-specified tips by a fixed hard-coded
  tip. (Sections~\ref{ss:forward} and~\ref{ss:tradeoff})

\item EIP-1559's base fee update rule is somewhat arbitrary and should
  be adjusted over time. (Section~\ref{ss:update})

\item Variable-size blocks enable a new (but expensive) attack
  vector: overwhelm the network with a sequence of maximum-size
  blocks.  (Sections~\ref{sss:chooserate}--\ref{sss:elastic})

\end{enumerate}

\subsection{Organization of Report}

Section~\ref{s:tfm} reviews Ethereum's current transaction fee
mechanism and provides a detailed description of the changes proposed
in EIP-1559.  Section~\ref{s:market} considers the market for
computation on the Ethereum blockchain and the basic forces of supply
and demand at work.  Section~\ref{s:ux} formalizes the concepts of a
``good user experience'' and ``easy fee estimation'' via posted-price
mechanisms.  Section~\ref{s:defs} defines several desirable
game-theoretic guarantees at the time scale of a single block, and
Section~\ref{s:mm} delineates the extent to which the  transaction fee
mechanism proposed in EIP-1559 satisfies them.
Section~\ref{s:collusion} investigates the possibility of collusion by
miners over long time scales.  Section~\ref{s:alt} spells out the
fatal flaws with some natural alternative designs and identifies 
worthy directions for further design experimentation.
Section~\ref{s:add} covers
additional benefits of
the mechanism proposed in EIP-1559, along with a short discussion of EIP-2593
(the ``escalator'').
Section~\ref{s:conc} concludes.

Sections~\ref{s:tfm}--\ref{s:ux}, \ref{s:collusion}, and
\ref{s:add}--\ref{s:conc} are relatively non-technical and meant for a
general audience.  Sections~\ref{s:defs}--\ref{s:mm} and~\ref{s:alt}
are more mathematically intense and aimed at readers who have at
least a passing familiarity with mechanism design theory (see
e.g.~\cite{f13} for the relevant background).\footnote{Other economic
  analyses of EIP-1559 include~\cite{aijan,1559analysis,pintail,qureshi,MZ51}.}

\section{Transaction Fee Mechanisms in Ethereum: Present and Future}\label{s:tfm}

This section reviews the economically salient properties of Ethereum
transactions (Section~\ref{ss:tx}), the status quo of a first-price
transaction fee mechanism (Section~\ref{ss:fpa}), the nuts and bolts
of the new transaction fee mechanism proposed in EIP-1559
(Section~\ref{ss:1559}), and the intuition behind the proposal
(Section~\ref{ss:1559intuition}).

\subsection{Transactions in Ethereum}\label{ss:tx}

The Ethereum blockchain, through its Ethereum virtual machine (EVM), 
maintains state (such as account balances) and carries out
instructions that change this state (such as transfers of the native
currency, called ether
(ETH)).  A {\em transaction} specifies a sequence of instructions to
be executed by the EVM.  The creator of a transaction is responsible
for specifying, among other fields, a {\em gas limit} and a {\em gas
  price} for the transaction.  The gas limit is a measure of the cost
(in computation, storage, and so on) imposed on the Ethereum
blockchain by the transaction.  The gas price specifies how much the
transaction creator is willing to pay (in ETH) per unit of gas.  For
example, the most basic type of transaction (a simple transfer) requires
21,000 units of gas; more complex transactions require more gas.
Typical gas prices reflect the current demand for EVM computation and
have varied over time by orders of magnitude; readers wishing to keep a
concrete gas price in mind could use, for example, 100 gwei (where one
gwei is~$10^{-9}$ ETH).
The total amount that the creator of
a transaction offers to pay for its execution is then the gas limit
times the gas price:
\begin{equation}\label{eq:bid}
\text{amount paid} := \text{gas limit} \times \text{gas price}.
\end{equation}
For example, for a 21,000-gas transaction with a gas price of 100
gwei, the corresponding payment would be $2.1 \times 10^{-3}$ ETH (or
1.26 USD at an exchange rate of 600 USD/ETH).

A {\em block} is an ordered sequence of transactions and associated
metadata (such as a reference to the predecessor block).  There is a
cap on the total gas consumed by the transactions of a block, which we
call the {\em maximum block size}.  The maximum block size
has increased over time and is currently~12.5M gas, enough
for roughly 600 of the simplest transactions.  Blocks are created and
added to the blockchain by {\em miners}.  
Each miner maintains a {\em mempool} of outstanding transactions and
collects a subset of them into a block.  To add a block to the
blockchain, a miner provides a proof-of-work in the form of a solution
to a computationally difficult cryptopuzzle; the puzzle
difficulty is adjusted over time to maintain a target rate of block
creation (roughly one block per 13 seconds).  Importantly, the miner
of a block has dictatorial control over which outstanding transactions
are included and their ordering within the block.  Transactions are
considered confirmed once they are included in a block that is added
to the blockchain.  The current state of the EVM is then the result of
executing all the confirmed transactions, in the order they appear
in the blockchain.\footnote{Technically, a longest-chain rule is used
  to resolve forks (that is, two or more blocks claiming a common
  predecessor). The confirmed transactions are then defined as those
  in the blocks that are well ensconced in the longest chain (that is,
  already extended by sufficiently many subsequent blocks).}

The {\em transaction fee mechanism} is the part of the protocol that
determines the amount that a creator of a confirmed transaction 
pays, and to whom that payment is directed.  

\subsection{First-Price Auctions}\label{ss:fpa}

Ethereum's transaction fee mechanism is and always has been a {\em
  first-price auction}~\cite{ethereum}.\footnote{First-price auctions
  are also used in Bitcoin~\cite{bitcoin}.}
\begin{mdframed}[style=offset,frametitle={First-Price Auctions}]
\begin{enumerate}

\item {\em Who pays what?}  The creator of a confirmed transaction
  pays the specified gas limit times the specified gas
  price (as in~\eqref{eq:bid}). 

\item {\em Who gets the payment?}  The entire payment is transferred
  to the miner of the block that includes the transaction.\tablefootnote{We will
    ignore details concerning transactions that run out of gas or
    complete with unused gas.}

\end{enumerate}
\end{mdframed}

A user submitting a transaction 
is sure to pay either the amount in~\eqref{eq:bid}
(if the transaction is confirmed) or~0 (otherwise).  A miner who mines
a block is sure to receive as revenue 
the amount in~\eqref{eq:bid} from each of the transactions it chooses
to include.  Accordingly,
many miners pack blocks up to 
the maximum block size, greedily prioritizing the outstanding
transactions with the highest gas prices.\footnote{Technically, because
  different transactions
  have different gas limits, selecting the revenue-maximizing set of
  transactions is a knapsack problem (see e.g.~\cite{ai3}).  The
  minor distinction between optimal and greedy knapsack solutions is
  not important for this report.}\fnsep\footnote{We use the word ``greedy''
  without judgment---``greedy algorithm'' is a standard term for a
  heuristic that is
  based on a   sequence of myopic decisions.}

\subsection{EIP-1559: The Nuts and Bolts}\label{ss:1559}

\subsubsection{Burning a History-Dependent Base Fee}

EIP-1559, following Buterin~\cite{vb14,vb16,vb1559}, proposes a mechanism that
makes several tightly coupled changes to the status quo.
\begin{mdframed}[style=offset,frametitle={EIP-1559: Key Ideas (1--3 of~8)}]
\begin{enumerate}

\item Each block has a protocol-computed reserve price (per unit of
  gas) called the {\em base fee}.  Paying the base fee is a
  prerequisite for inclusion in a block.\tablefootnote{Technically,
    a miner can also include a transaction unwilling to pay the full
    base fee, but it must then dip into its block reward to make up the difference.
    We ignore this detail in this report.}

\item The base fee is a function of the preceding blocks only, and
does not depend on the transactions included in the current
block.

\item All revenues from the base fee are burned---that is, permanently
  removed from the circulating supply of ETH.

\end{enumerate}
\end{mdframed}
Removing ETH from the supply increases the value of every
ether still in circulation.  
Fee-burning can therefore be viewed as a lump-sum
refund to ETH holders (\`{a} la stock buybacks).

The second point is underspecified; how, exactly, is the base fee
derived from the preceding blocks?  Intuitively, increases and
decreases in demand should put upward and downward pressure on the
base fee, respectively.\footnote{In the economics literature, such
  demand-dependent price adjustment is called ``t\^{a}tonnement''
  (French for ``groping'').}
But the blockchain records only the confirmed
transactions, not the transactions that were priced out.  If miners
publish a sequence of full (12.5M gas) blocks, how can the protocol
distinguish whether the current base fee is too low or exactly
right?

\subsubsection{Variable-Size Blocks}

The next key idea is to relax the constraint that every block has size
at most 12.5M gas and instead require only that the {\em average} block
size is at most 12.5M gas.\footnote{More generally, EIP-1559 is
  parameterized by a
  target block size, which is adjusted by miners over time 
(like the maximum block size is now).
For concreteness, throughout this report we assume a target block
size of 12.5M gas, the current maximum block size.}
The mechanism in EIP-1559 then uses past block sizes as an on-chain
measure of demand, with big blocks (more than 12.5M gas) and small
blocks (less than 12.5M gas) signaling increasing and decreasing demand,
respectively.\footnote{The flexibility provided by variable block
  sizes can also
  reduce the variance in equilibrium transaction fees and the delays
  experienced by some users; see Section~\ref{ss:toohigh}.}  
Some finite maximum block size is still needed to control network
congestion; the current EIP-1559 spec~\cite{1559spec} proposes using
twice the average block size.
\begin{mdframed}[style=offset,frametitle={EIP-1559: Key Ideas (continued)}]
\begin{enumerate}

\item [4.] Double the maximum block size (e.g., from 12.5M gas
  to 25M gas), with the old maximum (e.g., 12.5M gas) serving as the
  {\em target} block size.

\item [5.] Adjust the base fee upward or downward whenever the
  size of the latest block is bigger or smaller than the target
  block size, respectively.

\end{enumerate}
\end{mdframed}
The specific adjustment rule proposed in the EIP-1559
spec~\cite{1559spec} computes the base
fee $r_{cur}$ for the current block from the base fee $r_{pred}$ and
size $s_{pred}$ of the predecessor block using the following formula,
where~$s_{target}$ denotes the target block size:\footnote{For
  simplicity, we ignore numerical details such as rounding the base
  fee to an integer.}
\begin{equation}\label{eq:update}
r_{cur} := r_{pred} \cdot \left(1 + \frac{1}{8} \cdot \frac{s_{pred} -
    s_{target}}{s_{target}} \right).
\end{equation}
For example, the base fee increases by 12.5\% after a maximum-size
block (i.e., double the target size) and decreases by 12.5\% after an
empty block.  A maximum-size block followed by an empty block (or vice
versa) leaves the base fee at $\tfrac{9}{8} \cdot \tfrac{7}{8} =
\tfrac{63}{64} \approx 98.4\%$ of its prior value.\footnote{See
also Table~\ref{table:basefee} in Section~\ref{sss:benefits} for a more
  complex example of this update rule in action, Monnot~\cite{monnot}
  for detailed simulations, and Filecoin~\cite{filecoin} for a recent
  deployment.}

If the base fee is burned rather than given to miners, why should
miners bother to include any transactions in their blocks at all?
Also, what happens when there are lots of transactions (more than 25M
gas worth) willing to pay the current base fee?  

\subsubsection{Tips}

The transaction fee mechanism proposed in EIP-1559 addresses the
preceding two questions by allowing the creator of a transaction to
specify a {\em tip}, to be paid above and beyond the base fee, which
is transferred to the miner of the block that includes the transaction
(as in a first-price auction).  Small tips should be sufficient to
incentivize a miner to include a transaction during a period of stable
demand, when there is room in the current block for all the
outstanding transactions that are willing to pay the base fee.  Large
tips can be used to encourage special treatment of a transaction, such
as a specific positioning within a block, or the immediate inclusion
in a block in the midst of 
a sudden demand spike.
\begin{mdframed}[style=offset,frametitle={EIP-1559: Key Ideas (continued)}]
\begin{enumerate}

\item [6.] Rather than a single gas price, a transaction now includes
a {\em tip} and a {\em fee cap}.  
A transaction will be included in a block only if its fee cap is at
least the block's base fee.

\item [7.] {\em Who pays what?} If a transaction with tip~$\delta$, fee cap $c$, and gas
  limit~$g$ is included in a block with base fee $r$, the transaction
  creator pays
  $g \cdot \min\{ r+\delta, c\}$ ETH.

\item [8.] {\em Who gets the payment?}  
Revenue from the base fee (that is, $g \cdot r$) is burned and the
remainder ($g \cdot \min\{ \delta, c-r\}$) is transferred to the miner
of the block.

\end{enumerate}
\end{mdframed}
For example, consider a block with base fee~100 (in gwei per unit of
gas).  If the
block's miner includes a transaction with tip~4 and fee cap~200, the
creator of that transaction will pay~104 gwei per unit of gas (100 of which
is burned, 4 of which goes to the miner).  An included transaction
with tip~10 and fee cap~105 would pay~105 gwei per unit of gas (100 of
which is burned, 5 of which goes to the miner).

A user submitting a transaction with tip~$\delta$ and fee cap~$c$ is
sure to pay at most~$c$ gwei per unit of gas, and will pay less whenever
the current base fee is small (i.e., less than $c-\delta$).  A miner
who mines a block is sure to receive all the revenue from the tips of
the transactions it chooses to include.  Accordingly, one might expect
a typical miner to include all the transactions with fee cap greater
than the base fee.
If the total gas consumed by such transactions exceeds the maximum
block size of 25M gas, one might expect the miner to pack its block
full, greedily prioritizing the outstanding transactions with the
highest tips.

\subsection{An Informal Argument for EIP-1559}\label{ss:1559intuition}

The number of new ideas in EIP-1559 can be overwhelming.  Why so many
changes at once?  Does one of the changes necessitate the rest?  We
next outline one narrative of why EIP-1559 might have to look more or
less the way that it does, taking as given the goal of making fee
estimation far easier for users than in the status quo.  
The remainder of this report will interrogate this narrative
mathematically and explore some alternative designs.
\begin{mdframed}[style=offset,frametitle={Why EIP-1559 Looks the Way
    That It Does (Informal Argument)}]
\begin{enumerate}

\item First-price auctions are challenging for users to reason about
  because a user's optimal gas price depends on the gas prices offered
  by other users at the same time.

\item Other common auction designs in which the prices charged depend
  on the set of included transactions, such as second-price (a.k.a.\
  Vickrey) auctions, can be easily manipulated by miners through fake transactions.

\item Simple fee estimation, in which users are not forced to
  reason about other users' behavior, therefore seems to require a
  base fee---a price that is set independently of the 
  transactions included in the current block.

\item The ideal base fee would result in blocks filled with the
  highest-value transactions.  Demand changes over time, so
  the base fee must respond in kind.

\item The base fee revenues of a block must be burned or otherwise
  withheld from the block's miner, as otherwise the miner could
  collude with users off-chain to costlessly simulate a first-price
  auction.

\item 
Because demand is not recorded on-chain, an on-chain proxy
such as variable block sizes must be used to adjust the base fee.

\item Tips are required to disincentivize miners from publishing empty
  blocks.

\item Tips should be specified by users rather than hard-coded into
  the protocol so that high-value transactions can be identified and
  accommodated during a sudden demand spike.

\item Burning any portion of the tips would drive the tip market
  off-chain, and thus tips may as well be transferred entirely to a
  block's miner.

\end{enumerate}
\end{mdframed}

\section{The Market for Ethereum Transactions}\label{s:market}

This section steps away from the discussion of specific mechanisms and
focuses instead the basic forces of supply and demand at work in the
Ethereum blockchain.  Section~\ref{ss:market} defines a
``market-clearing outcome'' and posits it as the ideal outcome of a
transaction fee mechanism.  Section~\ref{ss:toohigh} emphasizes
that no mechanism can guarantee low transaction fees during periods in
which the demand for EVM computation significantly outstrips its
supply, and clarifies EIP-1559's likely effect on high transaction fees.

\subsection{Market-Clearing Prices and Outcomes}\label{ss:market}

The 12.5M gas available in an Ethereum block is a scarce resource, and
in a perfect world it should be allocated to the transactions that
derive the most value from it.  We can make this idea precise using a
{\em demand curve}, which is a decreasing function that specifies the
total amount of gas demanded by users at a given gas
price.\footnote{For simplicity of
  analysis, throughout this report we assume that demand is
  exogenous and independent of the choice of or actions by a
  transaction fee mechanism.  
Houy~\cite{houy} and Rizun~\cite{rizun} use a similar
formalism to reason about blockchain transaction fee markets.  
Richer models of demand, with pending transactions excluded from one
block persisting to the next, are studied by
Monnot~\cite{monnot,1559sim} in the context of EIP-1559 simulations
and by Easley et al.~\cite{EOB17} and Huberman et al.~\cite{leshno} to
carry out an economic analysis of Bitcoin's transaction fee mechanism.}
For example,
a {\em linear} demand curve has the form $D(p) = \max\{0, b - ap\}$,
where~$p$ denotes the gas price and $a,b \ge 0$ are nonnegative
constants (Figure~\ref{f:demand_curve}).

\begin{figure}
  \centering
    \includegraphics[width=0.5\textwidth]{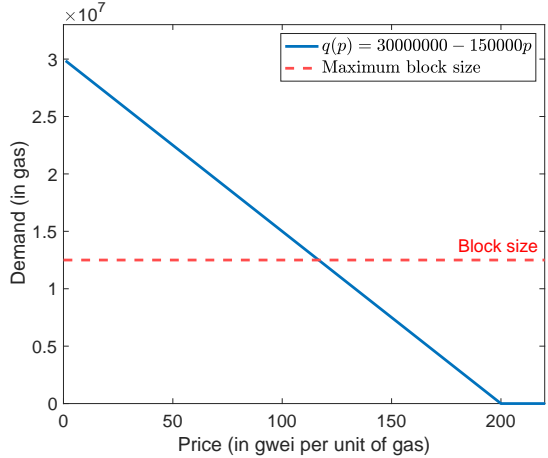}
\caption{An example linear demand curve, with $b = 30\text{M}$ and $a=150\text{K}$.
For example, there is a demand of 30M gas at a gas price of~0; zero
demand at a gas price of~200 gwei; and a demand of 12.5M gas at a gas
price of $116 \tfrac{2}{3}$ gwei.}
\label{f:demand_curve}
\end{figure}

The {\em market-clearing price} is then the price at which the total
amount of gas demanded equals the available supply (i.e., 12.5M gas).
For example, in Figure~\ref{f:demand_curve}, the market clearing price
is~$116 \tfrac{2}{3}$ gwei.  If the demand at price~0 is less than the
supply, we define the market-clearing price as~0.

The market-clearing price is the ideal gas price for a block.
For suppose such a
price~$p^*$ fell magically from the sky and became common knowledge to
all users, with the understanding that all confirmed transactions in
the current block will pay~$p^*$ per unit of gas.  In the resulting
outcome---the {\em market-clearing outcome}---users with maximum
willingness to pay at least~$p^*$ per unit of gas will opt to have
their transactions included, while those with a lower willingness to
pay opt out.  The end result?  The supply of 12.5M gas will be fully
utilized (because~$p^*$ is a market-clearing price), and moreover will
be allocated precisely to the highest-value transactions (those
willing to pay a gas price of at least~$p^*$).\footnote{Or if the
  supply constraint is not binding (and hence $p^*=0$),
  all transactions are included.}  Put differently, the
market-clearing outcome maximizes the value of the current block,
subject to the supply constraint of 12.5M gas.  For this reason, we
adopt the market-clearing outcome as the most desirable one for a
transaction fee mechanism.
\begin{mdframed}[style=offset,frametitle={Ideal Outcome of a
    Transaction Fee Mechanism}]
Every block is fully utilized by the highest-value transactions, with
all transactions paying a gas price equal to the market-clearing
price.
\end{mdframed}
Both the status quo and EIP-1559 transaction fee mechanisms can be
viewed as striving for this ideal, market-clearing outcome.
In first-price auctions, users are expected to estimate
what the current market-clearing price might be and bid accordingly.
In the EIP-1559 mechanism, the protocol continually adjusts the base
fee in search of the market-clearing price.

\begin{remark}[Revenue as a Necessary Evil]\label{rem:revenue}
The purpose of the market-clearing price is to differentiate high-value
and low-value transactions, so that the scarce resource that is an
Ethereum block can be allocated in the most valuable way.  Revenue is
generated in the market-clearing outcome (provided the supply
constraint is binding), but only as a side effect in the service of
economic efficiency.  The revenue-maximizing price is generally higher
than the market-clearing price, and it plays an important role in the
discussion in Section~\ref{s:collusion} of possible attacks by
colluding miners.
\end{remark}

\begin{remark}[Non-Zero Marginal Costs]\label{rem:marg_costs}
The preceding definition of a market-clearing outcome assumes that the
marginal cost to a miner of including an additional transaction in its
block is~0 (or $+\infty$, if including the transaction would
violate the cap of 12.5M gas).  In reality, every transaction imposes 
a small marginal cost on the miner; for example, one factor is that
the probability that a block is orphaned from the main chain (i.e.,
the ``uncle rate'') increases with the block
size~\cite{DW13}.

If the overall marginal cost to a miner is~$\mu$ gwei per unit of gas, then
$\mu$ plays the role of~0 in the more general definitions of
market-clearing prices and outcomes.\footnote{Alternatively, $\mu$ is
  the minimum compensation per unit of gas that a miner is willing to
  accept for including a transaction.}
That is, if the demand at
price~$\mu$ is at most the supply of 12.5M gas, the market-clearing
price is~$\mu$; in the corresponding outcome, all transactions
willing to pay a gas price of at least $\mu$ are included in the
block.
\end{remark}

\subsection{Will EIP-1559 Lower Transaction Fees?}\label{ss:toohigh}

The Ethereum community is justifiably concerned about overly high
transaction fees crowding out all but the most lucrative uses of the
Ethereum blockchain (e.g., DeFi arbitrage opportunities).  
No transaction fee mechanism can be a panacea to this problem.
This section clarifies what effects on transaction fees should and
should not be expected from the adoption of the transaction fee
mechanism proposed in EIP-1559.

\subsubsection{The Problem of High Market-Clearing Prices}\label{sss:toohigh}

First, whatever the mechanism, real transaction fees cannot be
expected to drop significantly below the market-clearing price during
a period of relatively stable demand.  With fees below that price,
demand for gas would exceed supply, resulting in some lower-value
transactions replacing higher-value transactions.  For example, with
the demand curve in Figure~\ref{f:demand_curve}, if typical fees
dropped to~100 gwei per unit of gas, the demand would be~15M gas.  The
2.5M gas worth of excluded transactions will inevitably include some
for which the creator's willingness to pay is at least the
market-clearing price of~$116 \tfrac{2}{3}$ gwei.  Such users should
be expected to push up transaction fees and guarantee inclusion of
their transactions, either on-chain through the transaction fee
mechanism (e.g., by increasing a transaction's gas price in a
first-price auction), or off-chain through a side agreement with a
miner.

But what if the
market-clearing price is already unacceptably high?
The only ways to decrease the market-clearing price are to increase
supply or decrease demand (Figure~\ref{f:shift})---actions that are
generally outside the purview of mechanism design.
\begin{mdframed}[style=offset,frametitle={Scalability vs.\ Mechanism Design}]
Lowering the market-clearing price by increasing supply or decreasing
demand is fundamentally a scalability problem, not a mechanism design
problem.  
\end{mdframed}
For example, typical layer-1 scaling solutions like
sharding,
in which different parts of the blockchain operate in
parallel, increase transaction throughput and therefore decrease the
market-clearing price.  Typical layer-2 scaling solutions like payment
channels
and rollups, which
effectively move some transactions off-chain, decrease demand for EVM
computation and likewise decrease the market-clearing price.
Looking toward the near future, good scaling solutions will be crucial
for keeping transaction fees in check and more generally for
encouraging the growth of the Ethereum network.

\begin{figure}[h]
  \centering
  \subfigure[Increasing the supply]{%
    \includegraphics[width=0.475\textwidth]{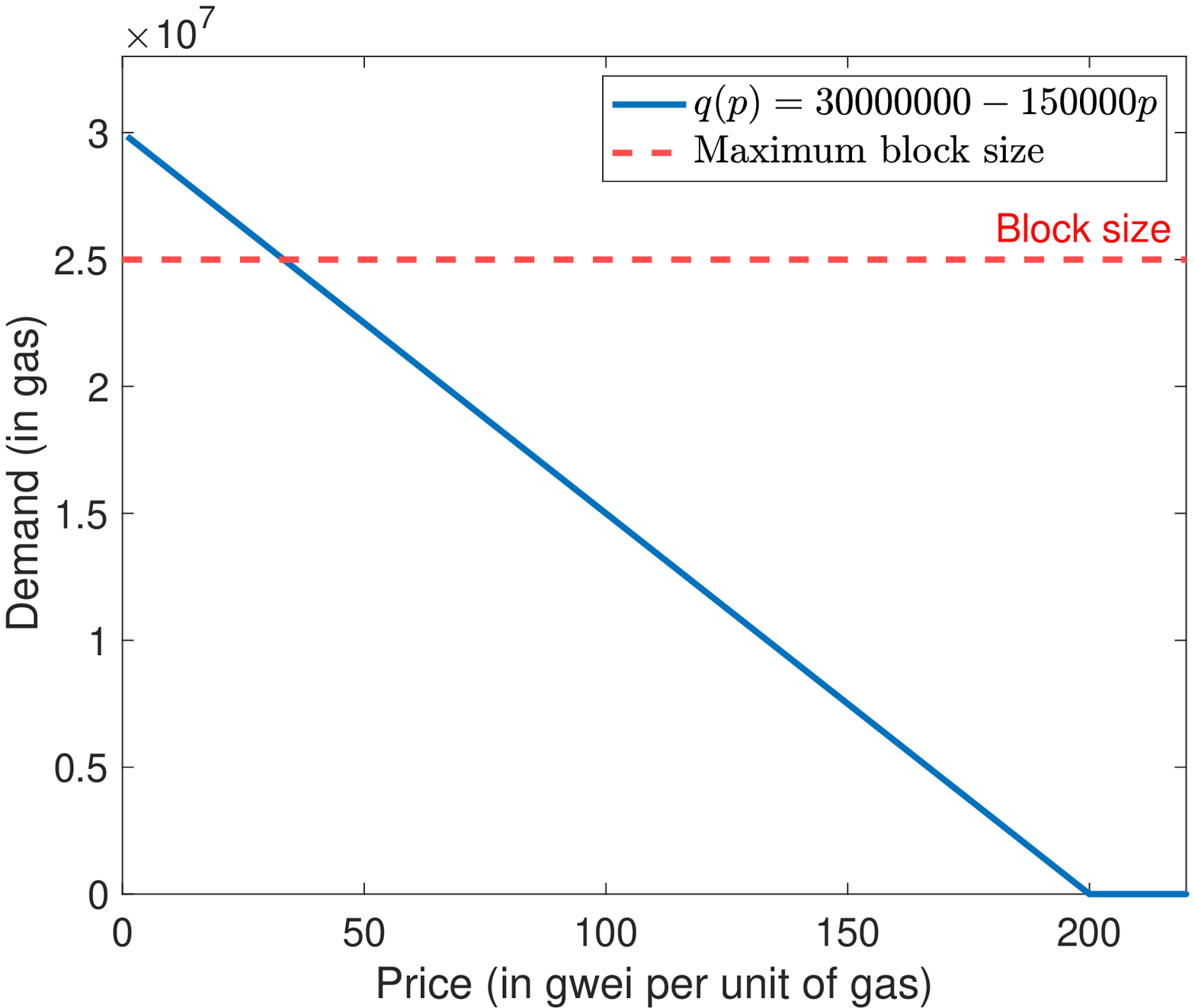}}
\quad
  \subfigure[Decreasing the demand]{%
    \includegraphics[width=0.475\textwidth]{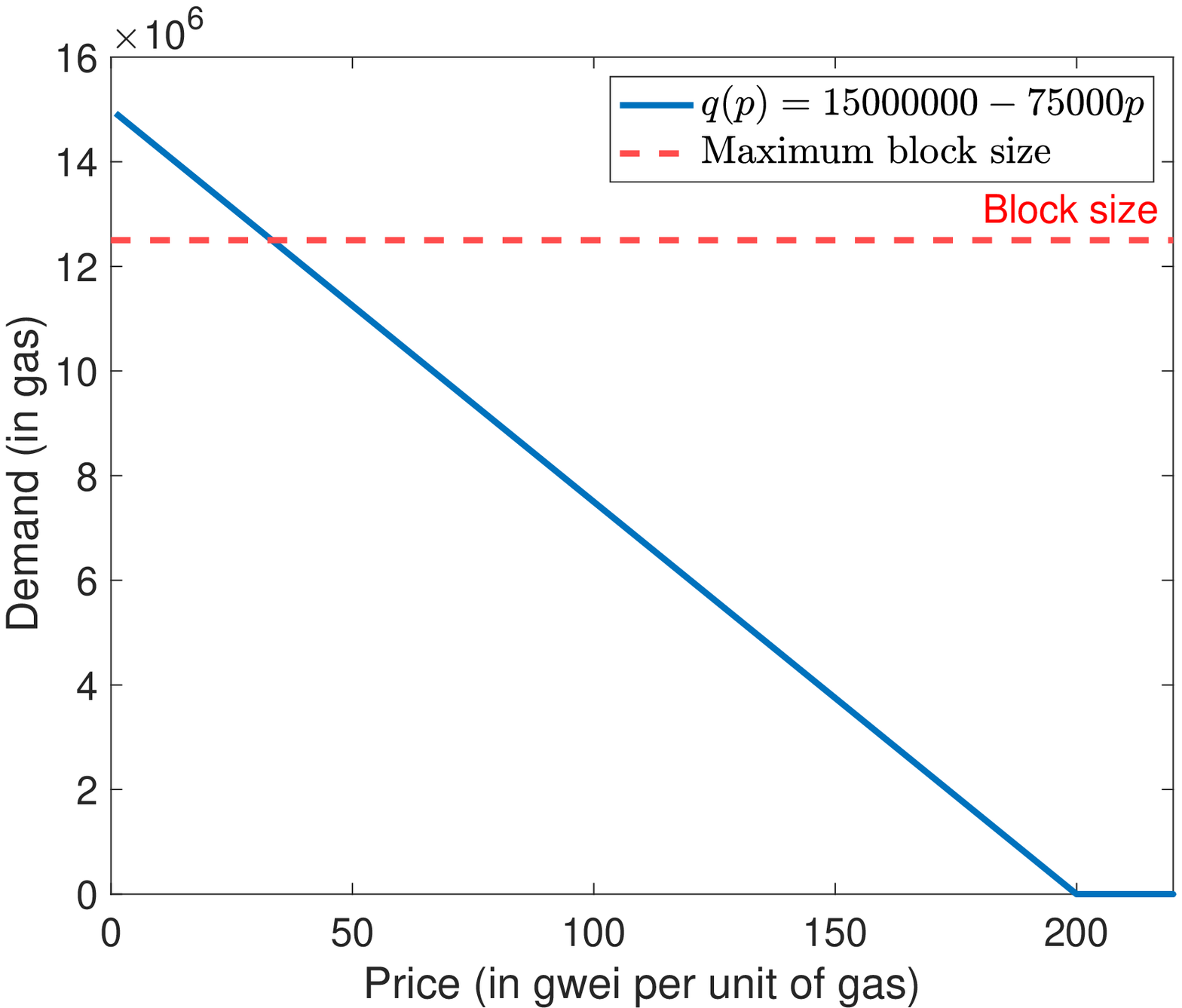}}
\caption{In the example in Figure~\ref{f:demand_curve}, doubling the
  supply (shown in~(a)) or halving the demand (shown in~(b)) cuts the
  market-clearing price from~$116 \tfrac{2}{3}$ gwei to~$33
  \tfrac{1}{3}$ gwei.}\label{f:shift}
\end{figure}

\subsubsection{Two Potential Benefits of EIP-1559}\label{sss:benefits}

The transaction fee mechanism proposed in EIP-1559 
has the potential to partially mitigate high transaction fees in two
different ways.  First, in a period of relatively stable demand, 
users can adopt the base fee as a good known-in-advance proxy for the
market-clearing price; this should lead to less guesswork and
consequent overpayment than in today's first-price auctions.
See also the discussion in Section~\ref{ss:feeest}.

Second, in a period of volatile demand, the mechanism proposed in
EIP-1559 can reduce the {\em 
  variance} in transaction fees experienced by users by exploiting
variable block sizes---in effect, borrowing capacity from the near
future to use in a time of need.  
This flexibility in block sizes can reduce the maximum
transaction fee paid during the period (as well as the delay
experienced by some users).
\begin{example}[Trajectory of EIP-1559]\label{ex:traj}
Consider the trajectory of the EIP-1559 mechanism that is
detailed in Table~\ref{table:basefee} and depicted in
Figure~\ref{f:basefee}.  For this example, we assume that tips are
negligible and that a transaction is included in a block if and only
if its fee cap is at least the current base fee.  Period~1 represents
the end of a long era of stable demand, during which the base fee 
converged to the market-clearing price for the
target block size (12.5M gas).  
Demand doubles for the next six periods.
With a fixed supply of 12.5M gas, the market-clearing price
jumps suddenly from $33\tfrac{1}{3}$ to~$116\tfrac{2}{3}$ after
period~1, and back to~$33 \tfrac{1}{3}$ after period~7.  In the
EIP-1559 mechanism, the base fee---the mechanism's guess at the
current market-clearing price for the target block size---increases
slowly but surely, with larger-than-target blocks absorbing the excess
demand along the way.  Once demand returns to its original level,
blocks will have size smaller than the target as the mechanism's base
fee slowly but surely decreases to the new market-clearing price.  
In this example, the
maximum base fee of
61.69 (in period~8) is only about 53\% of the maximum market-clearing
price with a fixed block size of 12.5M gas ($116\tfrac{2}{3}$, in
periods~2--7).
\end{example}

\begin{table}
\centering
{\footnotesize
\begin{tabular}{|c|c|c|c|c|c|c|c|c|}\hline
& Period 1 & Period 2 & Period 3 & Period 4 & Period 5 & Period 6 & Period 7  & Period 8\\ \hline\hline
Demand & Low & High & High & High & High & High & High & Low \\ \hline
M-C Price (12.5M) & $33.33$
& $116.67$ & $116.67$ & $116.67$ & $116.67$ & $116.67$ & $116.67$ & $33.33$\\ \hline
EIP-1559 Base Fee 
& $33.33$ & $33.33$ & $37.5$ & $41.95$ & $46.65$ & $51.55$ & $56.59$ & $61.69$
\\ \hline
EIP-1559 Block Size 
& 12.5M & 25M & 24.38M & 23.71M & 23M & 22.27M & 21.51M & 10.37M
\\ \hline
\end{tabular}
}
\caption{An example of the EIP-1559 base fee adjustment rule in
  action.  ``Low'' demand means the demand curve
  $D(p)=15000000-75000p$ shown in Figure~\ref{f:shift}(b); ``high''
  means the demand curve $D(p)=30000000-150000p$ shown in
  Figure~\ref{f:demand_curve}.  (Here ``demand'' means the total gas
  consumed by all pending transactions that have a fee cap of~$p$ or
  more.)
The second row shows the market-clearing price for each demand curve
when the supply is fixed at 12.5M gas.  The third and fourth rows show
the joint evolution of the base fee and block size under the EIP-1559 mechanism, assuming that the base fee matches the market-clearing price in period~1 and that all users submit negligible tips.
}\label{table:basefee}
\end{table}

\begin{figure}[h]
  \centering
  \subfigure[Price comparison]{%
    \includegraphics[width=0.475\textwidth]{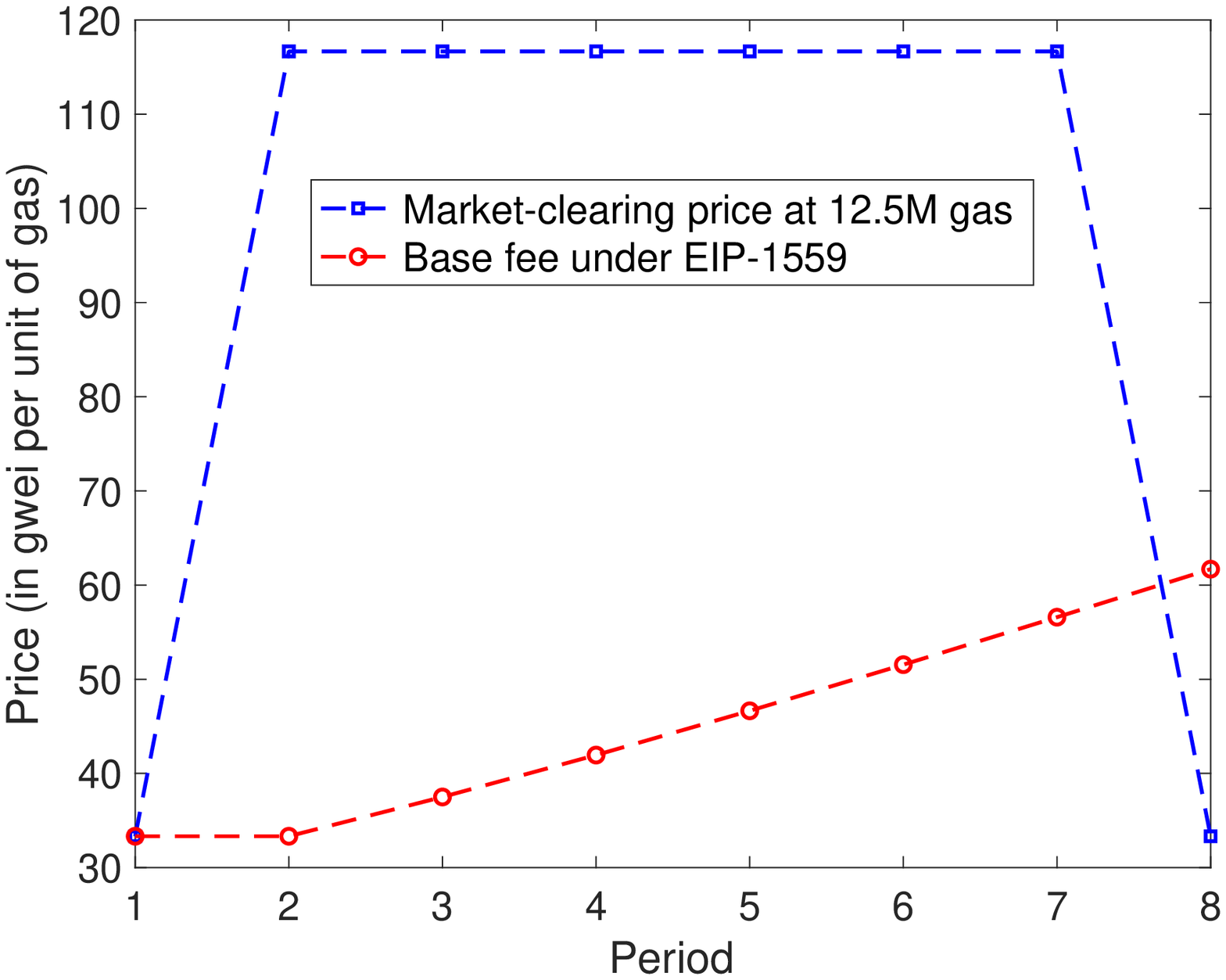}}
\quad
  \subfigure[Block size comparison]{%
    \includegraphics[width=0.475\textwidth]{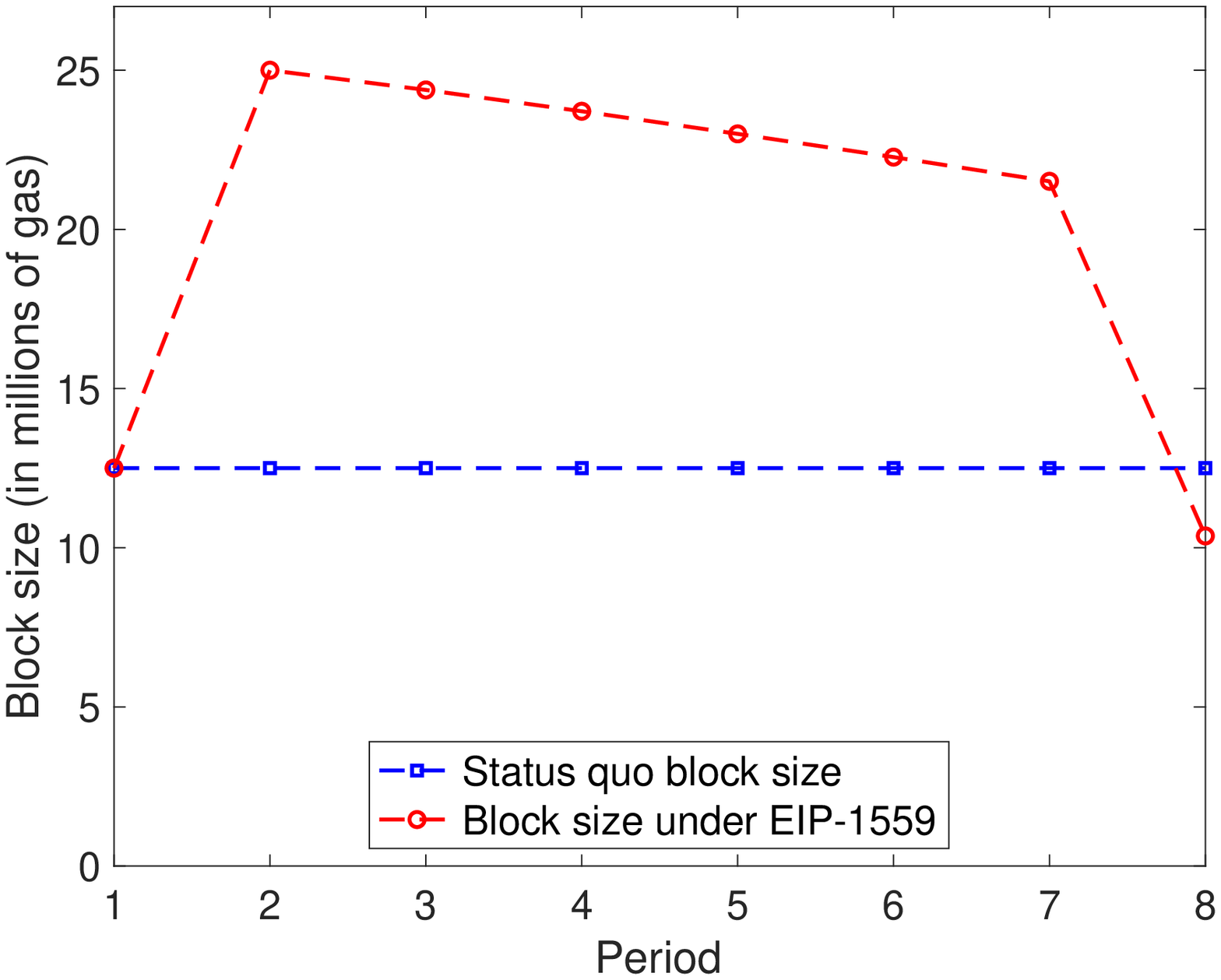}}
\caption{For the example detailed in Table~\ref{table:basefee}, a comparison of
  the price and block size under the status quo and under EIP-1559.
  In the subsequent periods, the base fee and block size under
  EIP-1559 gradually return to $33 \tfrac{1}{3}$ gwei and 12.5M gas,
  respectively.}\label{f:basefee}
\end{figure}

\section{The Purpose of EIP-1559: Easy Fee Estimation}\label{s:ux}

\subsection{The Problem of Fee Estimation}\label{ss:feeest}

With or without EIP-1559, transaction fees will be high whenever the
demand for EVM computation far exceeds its supply
(Section~\ref{ss:toohigh}).  So what's the point of the proposal?
To make transaction fees {\em more predictable} and thereby make the
fee estimation problem---the problem of choosing the optimal gas
price for a transaction---as straightforward as possible.

Ethereum users appear to overpay regularly for EVM computation,
offering gas prices that are significantly larger than the
market-clearing price~\cite{etherscan}.  Part of the problem
may be attributable to poor fee estimation algorithms in wallets, 
which could conceivably improve over time (see e.g.~\cite{lopp,newbery}, in
the similar context of Bitcoin).
But part of the problem is fundamental to first-price auctions, and
addressing it necessitates a major change in the transaction fee
mechanism.\footnote{Bidding in a first-price auction has long been
  known to be a hard problem; see e.g.~\cite{CST80}.}

\begin{mdframed}[style=offset,frametitle={EIP-1559: Improving the User
    Experience with Easy Fee Estimation}]
This report assumes that the primary purpose of EIP-1559 is to improve
the ``user experience (UX)'' of Ethereum users, and to do so
specifically by making the fee estimation problem as easy as possible.
\end{mdframed}
EIP-1559 also offers a number of other benefits (see
Section~\ref{ss:benefits}), which are treated in this report as happy
accidents---byproducts of the proposed UX improvements.\footnote{This
  viewpoint appears consistent with the original motivation for
  EIP-1559.  Buterin~\cite{vb_ec18} writes: ``Our goal is to
  discourage the development of complex miner strategies and complex
  transaction sender strategies in general, including both complex
  client-side calculations and economic modeling as well as various
  forms of collusion.''}

\subsection{Auctions vs.\ Posted-Price Mechanisms}\label{ss:postedprice}

To what extent does EIP-1559 achieve its goal of a ``better UX''?
``User experience'' is a vague term, and it must be defined
mathematically before this question can be answered.  

Definition~\ref{def:uic} presents our formalization of ``good UX,''
and the intuition for it is simple:
Shopping on Amazon is a lot easier than buying a house in a
competitive real estate market.
On Amazon, there's no need to be strategic or second-guess yourself;
you're either willing to pay the listed price for the listed product,
or you're not.  The outcome is economically efficient in that every
user who buys a product has a higher willingness to pay for it than
every user who doesn't buy the product.

When pursuing a house and competing with other potential buyers, you
must think carefully about what price to offer to the seller.  And no
matter how smart you are, you might regret your offer in
hindsight---either because you underbid and were outbid at a price you
would have been willing to pay, or because you overbid and paid more
than you needed to.  The house need not be sold to the potential buyer
willing to pay the most (if that buyer shades their bid too
aggressively), which is a loss in economic efficiency.

Bidding in Ethereum's first-price auctions is like buying a house.
Estimating the optimal gas price for a transaction requires 
making educated guesses about the gas prices chosen for the competing
transactions.
From a user's perspective, any bid may end up looking too high or too
low in hindsight.  From a societal perspective, lower-value transactions
that bid aggressively may displace higher-value transactions that
do not.

Could we redesign Ethereum's transaction fee mechanism so that setting
a transaction's gas price is more like shopping on Amazon?  Ideal
would be a {\em posted-price mechanism}, meaning a mechanism that
offers each user a take-it-or-leave-it gas price for inclusion in the
next block.  We'll see in Section~\ref{ss:1559uic} that the
transaction fee mechanism proposed in EIP-1559 acts like a
posted-price mechanism except when there is a large and sudden
increase in demand (Theorem~\ref{t:1559uic}).

\section{Incentive-Compatible Transaction Fee Mechanisms}\label{s:defs}

This section formalizes three desirable game-theoretic guarantees for
a transaction fee mechanism.  First, miners should be incentivized to
carry out the mechanism as intended, and strongly disincentivized from
including fake transactions (Section~\ref{ss:mmic}).  Second, 
the optimal gas price to specify should be obvious to the creator of a
transaction (Section~\ref{ss:uic}).  Finally, there should be no way
for miners and users to collude and strictly increase their utility by
moving payments off-chain (Section~\ref{ss:oca}).  Sections~\ref{ss:model}
and~\ref{ss:model2} set up the notation and language necessary to
formally state these three definitions.

This and the next section focus on incentives for miners and
users at the time scale of a single block, and on two important types
of attacks that can be carried out at this time scale (the insertion of fake
transactions, and off-chain agreements between miners and users).
Section~\ref{s:collusion} treats incentive issues and attacks that
manifest over longer time scales.

\subsection{The Basic Model}\label{ss:model}

On the supply side, let~$G$ denote the maximum size of a block in gas
(e.g., 12.5M gas in the status quo or 25M gas under EIP-1559), and $\mu \ge 0$
the marginal cost of gas to a miner (as in
Remark~\ref{rem:marg_costs}).\footnote{Equivalently, $\mu$ is the
  minimum gas price that a profit-maximizing miner is willing to
  accept in exchange for transaction inclusion when the maximum block
  size is not a binding constraint.  The
  formal definition of a ``profit-maximizing miner'' is given in Definition~\ref{def:mmutil}.}
For simplicity, we assume that~$\mu$ is the same for
all miners and common knowledge among users.\footnote{Calculations by
  Buterin~\cite{feemarketchange} suggest that~$\mu$ is, at this time of 
  writing, on the order of~0.4--3.3 gwei.  In a proof-of-stake
  blockchain such as ETH 2.0, the parameter~$\mu$ is likely to be even
  smaller.}
On the demand side, let~$M$ denote the set of transactions in
the mempool at the time of the current block's creation.

We associate three parameters with each transaction $t \in M$:
\begin{itemize}

\item a {\em gas limit} $g_t$ in gas;

\item a {\em value} $v_t$ in gwei per unit of gas;

\item a {\em bid} $b_t$ in gwei per unit of gas.

\end{itemize}
The gas limit is the amount of gas required to carry out the
transaction.  The value is the maximum gas price the transaction's
creator would be willing to pay for its execution in the current
block.\footnote{We assume that the value is independent of the
  position in the block, ignoring e.g.\ front-running bots aiming to
  secure the first position in a block (see~\cite{flashboys,RK20}).}
The bid corresponds to the gas
price that the creator actually offers to pay, which in general can be
less (or more) than the value.  With a first-price auction, the bid
corresponds to the gas price specified for a transaction.  In the
transaction fee mechanism proposed in EIP-1559, the bid corresponds to
the minimum of the fee cap
and the sum of the base fee and the tip ($\min \{ r+\delta, c\}$
in the notation of Section~\ref{ss:1559}).
We view the gas limit and value as
immutable properties of a transaction; the bid, by contrast, is under
control of the transaction's creator.
The gas limit and bid of a confirmed transaction are recorded
on-chain; the value of a transaction is known solely to its creator.

\subsection{Allocation, Payment, and Burning Rules}\label{ss:model2}

A transaction fee mechanism decides which transactions should be
included in the current block, how much the creators of those
transaction have to pay, and to whom their payment is directed.  These
decisions are formalized by three functions: an {\em allocation rule},
a {\em payment rule}, and a {\em burning rule}.

\subsubsection{Allocation Rules}

We use
$B_1,B_2,\ldots,B_{k-1}$ to denote the sequence of blocks in the
current longest chain (with~$B_1$ the genesis block and~$B_{k-1}$ the
most recent block) and~$M$ the pending transactions in the mempool.
Generally, bold type (like $\allocs$) will indicate a vector
and regular type (like $\alloct$) one of its components.

\begin{definition}[Allocation Rule]
An {\em allocation rule} is a vector-valued
function~$\allocs$ from the on-chain
history~$B_1,B_2,\ldots,B_{k-1}$ and mempool~$M$ to a
0-1 value $\alloct(B_1,B_2,\ldots,B_{k-1},M)$ for each pending  transaction~$t \in M$.
\end{definition}
A value of~1 for $\alloct(B_1,B_2,\ldots,B_{k-1},M)$
indicates transaction~$t$'s inclusion in the current
block~$B_k$; a value of~0 indicates its exclusion.  We sometimes
write~$B_k = \allocs(B_1,B_2,\ldots,B_{k-1},M)$, with the
understanding that~$B_k$ is the set of transactions~$t$ for which
$\alloct(B_1,B_2,\ldots,B_{k-1},M) = 1$.

We consider only feasible allocation rules, meaning allocation rules
that respect the maximum block size~$G$.
\begin{definition}[Feasible Allocation Rule]
An allocation rule $\allocs$ is {\em feasible} if, for every possible
history $B_1,B_2,\ldots,B_{k-1}$ and mempool~$M$,
\begin{equation}\label{eq:feasible}
\sum_{t \in M} g_t \cdot \alloct(B_1,B_2,\ldots,B_{k-1},M) \le G.
\end{equation}
\end{definition}
We call a set~$T$ of transactions {\em feasible} if they can all be
packed in a single block: $\sum_{t \in T} g_t \le G$.

\begin{remark}[Miners Control Allocations]\label{rem:minerscontrol}
While a transaction fee mechanism is generally designed with a specific
allocation rule in mind, it is important to remember that a miner
ultimately has dictatorial control over the block it creates.
\end{remark}
\begin{example}[First-Price Auction Allocation Rule]\label{ex:fpa_alloc}
The (intended) allocation rule~$\allocs^{f}$ in a first-price
auction is to include a feasible subset of outstanding transactions that
maximizes the sum of the gas-weighted bids, less the gas costs.
That is, the $\alloct^f$'s are assigned 0-1
values to maximize
\begin{equation}\label{eq:fpa_obj}
\sum_{t \in M} \alloct^{f}(B_1,B_2,\ldots,B_{k-1},M) \cdot (b_t-\mu)
  \cdot g_t,
\end{equation}
subject to~\eqref{eq:feasible}.
\end{example}

\subsubsection{Payment and Burning Rules}

The payment rule specifies the revenue earned by the miner from
included transactions.
\begin{definition}[Payment Rule]
A {\em payment rule} is a function~$\prices$ from the current on-chain
history $B_1 ,B_2, \ldots, B_{k-1}$ and transactions~$B_k$ included in
the current block to a nonnegative number
$\pricet(B_1,B_2,\ldots,B_{k-1},B_k)$ for each included transaction~$t
\in B_k$.
\end{definition}
The value of $\pricet(B_1,B_2,\ldots,B_{k-1},B_k)$ indicates the
payment from the creator of an included transaction~$t \in B_k$ to the
miner of the block~$B_k$ (in ETH, per unit of gas).

For example, in a first-price auction, a winner always pays its bid
(per unit of gas), no matter what the blockchain history and other
included transactions.
\begin{example}[First-Price Auction Payment Rule]\label{ex:fpa_payment}
In a first-price auction, 
\[
\pricet^f(B_1,B_2,\ldots,B_{k-1},B_k) = b_t
\]
for all~$B_1,B_2,\ldots,B_k$ and~$t \in B_k$.
\end{example}

Finally, the burning rule specifies the amount of ETH burned---or equivalently,
refunded to ETH holders---for each of the included transactions.
\begin{definition}[Burning Rule]
A {\em burning rule} is a function~$\burns$ from the current on-chain
history~$B_1,B_2,\ldots,B_{k-1}$ and transactions~$B_k$ included in
the current block to a nonnegative number
$\burnt(B_1,B_2,\ldots,B_{k-1},B_k)$ for each included transaction~$t
\in B_k$.
\end{definition}
The value of $\burnt(B_1,B_2,\ldots,B_{k-1},B_k)$ indicates the
amount of ETH burned (per unit of gas) by the creator of an included
transaction~$t \in B_k$.
\begin{example}[First-Price Auction Burning Rule]\label{ex:fpa_burn}
Status quo first-price auctions burn no fees, so
\[
\burnt^f(B_1,B_2,\ldots,B_{k-1},B_k) = 0
\]
for all~$B_1,B_2,\ldots,B_k$ and~$t \in B_k$.
\end{example}

\begin{remark}[The Protocol Controls Payments and Burns]
A miner does not control the payment or burning rule, except inasmuch as
it controls the allocation, meaning the transactions included
in~$B_k$.  Given a choice of allocation, the on-chain payments and fee
burns are completely specified by the protocol.  (Miners might seek out
off-chain payments, however; see Section~\ref{ss:oca}.)
\end{remark}

\begin{remark}[Mempool-Dependence]
The allocation rule~$\allocs$ depends on the mempool~$M$ 
because a miner can base its allocation
decision on the entire set of outstanding transactions.
Payment and burning rules must be computable from the on-chain
information~$B_1,B_2,\ldots,B_k$, and in particular cannot depend on
outstanding transactions of~$M$ excluded from the current block~$B_k$.
\end{remark}

\subsubsection{Transaction Fee Mechanisms}

Formally, a transaction fee mechanism is specified by its allocation,
payment, and burning rules.
\begin{definition}[Transaction Fee Mechanism (TFM)]
A {\em transaction fee mechanism (TFM)} is a triple $\tfm$ in which
$\allocs$ is a feasible allocation rule, $\prices$ is a payment rule,
and $\burns$ is a burning rule.
\end{definition}
For example, a first-price auction is mathematically encoded by the
triple $(\allocs^{f},\prices^{f},\burns^f)$ in which~$\allocs^f$ is the
revenue-maximizing allocation rule (Example~\ref{ex:fpa_alloc}),
$\prices^f$ is the pay-as-bid payment rule (Example~\ref{ex:fpa_payment}),
and $\burns^f$ is the all-zero burning rule (Example~\ref{ex:fpa_burn}).

Finally, we consider only individually rational mechanisms, meaning
TFMs that cannot force users to pay more than  their declared
willingness to pay.
\begin{definition}[Individual Rationality]
A TFM $\tfm$ is {\em individually rational} if, for every
history~$B_1,B_2,\ldots,B_k$,
\[
\underbrace{\pricet(B_1,B_2,\ldots,,B_k) +
  \burnt(B_1,B_2,\ldots,,B_k)}_{\text{total gas price paid by $t$'s creator}} \le b_t
\]
for every transaction~$t \in B_k$.
\end{definition}

\subsection{Incentive Compatibility (Myopic Miners)}\label{ss:mmic}

This section formalizes what it means for a TFM to be
game-theoretically sound from the perspective of miners---intuitively,
that a miner is incentivized to implement the intended allocation rule
and disincentivized from including fake transactions.
As a reminder, our current focus is on incentives at the time scale of
a single block, with longer time scales discussed in
Section~\ref{s:collusion}.

\subsubsection{Myopic Miner Utility Function}

In addition to choosing an allocation
(Remark~\ref{rem:minerscontrol}), we assume that miners can
costlessly add any number of fake transactions to the mempool (with
arbitrary gas limits and bids).
We call a miner {\em myopic} if its {\em utility}---meaning the
quantity that it acts to maximize---equals its net revenue from the
current block.\footnote{We ignore the block reward (currently~2 ETH),
as it is independent of the miner's actions and therefore irrelevant
for the single-block game-theoretic analysis in this and the next section.
The block reward does, of course, affect the security of the Ethereum blockchain (e.g.~\cite{auer,budish}).}
\begin{definition}[Myopic Miner Utility Function]\label{def:mmutil}
For a TFM $\tfm$, on-chain history~$B_1,$ $B_2,\ldots,B_{k-1}$, 
mempool~$M$, fake transactions~$F$, and choice~$B_k \sse M \cup F$ of
included transactions (real and fake), the utility of a {\em myopic
  miner} is
\begin{equation}\label{eq:mmutil}
u(F,B_k) := 
\underbrace{\sum_{t \in B_k \cap M} \pricet(B_1,B_2,\ldots,,B_k) \cdot
  g_t}_{\text{miner's revenue}}
-
\underbrace{\sum_{t \in B_k \cap F} \burnt(B_1,B_2,\ldots,,B_k) \cdot
  g_t}_{\text{fee burn for miner's fake transactions}}
- \underbrace{\mu \sum_{t \in B_k} g_t}_{\text{gas costs}}.
\end{equation}
\end{definition}
The first term sums over only the real included transactions, as for
fake transactions the payment goes from the miner to itself.  The
second term sums over only the fake transactions, as for real
transactions the burn is paid by their creators (not the miner).
In~\eqref{eq:mmutil}, we highlight the dependence of the utility
function on the two arguments that are under a miner's direct control,
the choices of the fake transactions~$F$ and included (real and 
fake) transactions~$B_k$.\footnote{We can assume that $F \sse B_k$,
as there's no point to creating and then excluding a fake
transaction.}

\subsubsection{Incentive-Compatibility for Myopic Miners}

A transaction fee mechanism is generally designed with a specific
allocation rule in mind (Remark~\ref{rem:minerscontrol}), but will
miners actually implement it?
\begin{definition}[Incentive-Compatibility for Myopic Miners (MMIC)]\label{def:mmic}
A TFM $\tfm$ is {\em incentive-compatible for myopic miners (MMIC)}
if, for every on-chain history~$B_1,B_2,\ldots,B_{k-1}$ and
mempool~$M$, a myopic miner maximizes its utility~\eqref{eq:mmutil} by
creating no fake transactions (i.e., setting $F=\emptyset$) and
following the suggestion of the allocation rule $\allocs$ (i.e.,
setting $B_k = \allocs(B_1,B_2,\ldots,B_{k-1},M)$).
\end{definition}
\begin{example}[First-Price Auctions Are MMIC]\label{ex:fpa_mmic}
A status quo first-price auction $\fpa$ is MMIC.  
Because $\burns^f$ is the all-zero function (Example~\ref{ex:fpa_burn}), 
the second term in~\eqref{eq:mmutil} is zero.   
Because payments equal bids
(Example~\ref{ex:fpa_payment}), miner utility equals the exact same
quantity~\eqref{eq:fpa_obj} maximized by the allocation rule $\allocs^f$
(Example~\ref{ex:fpa_alloc}).  Thus, myopic miner utility is maximized
by following the allocation rule and setting $B_k = \allocs^f(B_1,B_2,\ldots,B_{k-1},M)$.
\end{example}

\begin{example}[Vickrey (Second-Price) Auctions Are Not MMIC]\label{ex:spa_mmic}
{\em Vickrey} (a.k.a.\ {\em second-price}) auctions play as central a
role in
traditional auction theory as first-price auctions. Their claim to
fame is that, assuming the auction is implemented by a trusted third
party,  truthful bidding (i.e., setting one's bid $\bidt$ equal
to one's value $\valt$) is a dominant strategy, meaning it 
maximizes a bidder's utility no matter what the other bidders do.
This property sure sounds like ``easy fee estimation,'' so why not use
it as a TFM?

Unfortunately, Vickrey auctions can be manipulated via fake
transactions and thus fail to be MMIC.  For example, consider a set of
transactions that all have the same gas limit and a block that has
room for three of them.  In this setting, a Vickrey auction would
prescribe including the three transactions with the highest bids and
charging each of them (per unit of gas) the lowest of these three
bids.\footnote{Actually, a Vickrey auction would prescribe charging
  the highest losing bid rather than the lowest winning bid.  The
  former is off-chain and thus unimplementable in a blockchain
  context, while the latter is on-chain and typically close enough.}
Now imagine that the top three bids are 10, 8, and 3.  If a miner
honestly executes a Vickrey auction, its revenue will be $3 \times 3 =
9$.  If the miner instead submits a fake transaction with
bid~8 and executes a Vickrey auction (with the top two real
transactions included along with the fake transaction), its net
revenue jumps to $2 \times 8 = 16$.
\end{example}

\begin{remark}[Credible Mechanisms]\label{rem:credible}
The definition of MMIC (Definition~\ref{def:mmic}) is closely related
to Akbarpour and Li's notion of a {\em credible
  mechanism}~\cite{AL20}.  Intuitively, a mechanism is credible if the
agent tasked with carrying it out has no plausibly deniable
utility-improving deviation.  For instance,
Example~\ref{ex:spa_mmic} is a proof that the Vickrey auction is not
credible in this sense.  Akbarpour and Li~\cite{AL20} study both
single-shot (a.k.a.\ ``static'') mechanisms and mechanisms that
require many rounds (such as ascending auctions); the former type
are much more practical for blockchain transaction fee mechanisms.  
Interestingly, one of the main results in~\cite[Theorem 3.7]{AL20} is that
first-price auctions with an exogenously restricted bid space are the
only static credible mechanisms.\footnote{The results in~\cite{AL20}
  assume a computationally unbounded auctioneer.  Ferreira and
  Weinberg~\cite{FW20} explore what other credible mechanisms are
  possible assuming a computationally bounded auctioneer and the
  existence of cryptographically secure hash functions.}  
All of the MMIC mechanisms 
appearing in this report---first-price auctions
(Example~\ref{ex:fpa_mmic}), the 1559 mechanism (Theorem~\ref{t:1559mmic}),
and the tipless mechanism of Section~\ref{ss:tradeoff} (Theorem~\ref{t:tipless_mmic})---can be viewed
as first-price auctions with different restricted bid
spaces.\footnote{First-price auctions correspond to the bid
  space~$[0,\infty)$; 
the 1559 mechanism to the bid space $\{
  \text{``no bid''} \} \cup [r,\infty)$, where~$r$ is the block's base fee;
and the tipless mechanism to the bid space $\{
  \text{``no bid''}, r+\delta \}$, where~$r$ is the block's
  base fee   and~$\delta$ is a protocol-defined hard-coded tip.}
\end{remark}

Returning to status quo first-price auctions, the argument in
Example~\ref{ex:fpa_mmic} highlights two of their properties:
\begin{itemize}

\item [(i)] excluding real transactions suggested by the allocation
  rule strictly decreases myopic miner utility;

\item [(ii)] including fake transactions does not increase myopic
  miner utility.  

\end{itemize}
We next pursue a stronger version of property~(ii).

\subsubsection{$\cost$-Costly Transaction Fee Mechanisms}

A stronger version of property~(ii) would state that, as with
excluding real transactions, fake transactions significantly
decrease myopic miner utility.  
First-price auctions possess this
stronger property when the maximum block size constraint is binding
(as fake transactions then displace real ones)
or when the marginal cost~$\mu$ 
is large.
Otherwise, a miner can devote any extra room in a block to fake
transactions without suffering a significant cost.

The next definition formalizes this stronger version of property~(ii).

\begin{definition}[$\cost$-Costly Transaction Fee Mechanism]\label{def:costly}
A TFM $\tfm$ is {\em $\cost$-costly} if, for every on-chain
history~$B_1,B_2,\ldots,B_{k-1}$, mempool~$M$, fake
transactions~$F$, and block~$B_k \sse M \cup F$ chosen by a miner, the
fake transactions
of~$B_k$ decrease myopic miner utility~\eqref{eq:mmutil} by at least
$\cost$ per unit of gas:
\[
u(F,B_k) \le \underbrace{u(\emptyset,B_k \cap M)}_{\text{utility w/out
    fake txs}} - \underbrace{\cost \cdot \sum_{t \in F} g_t}_{\text{cost
    of fake txs}}.
\]
\end{definition}
For example, first-price auctions are $\mu$-costly, where~$\mu$ is the
marginal cost of gas to a miner, and are not $\cost$-costly for any
$\cost > \mu$.  
We'll see later (Corollary~\ref{cor:1559costly}) that the
transaction fee mechanism proposed in EIP-1559 is generally
$\cost$-costly for larger values of $\cost$, and in this sense more
aggressively punishes fake transactions.

\subsection{Incentive Compatibility (Users)}\label{ss:uic}

Next we formalize what it means for a TFM to be game-theoretically
sound from the perspective of users---intuitively, that there is
an ``obvious' optimal bid'' when creating a new transaction.
This is also our definition of a ``good user experience'' is the sense
of easy fee estimation (see Section~\ref{s:ux}).

\subsubsection{User Utility Function}

Recall from Section~\ref{ss:model} that the value~$\valt$ 
of a transaction~$t$ is the maximum gas price the
transaction's creator would be willing to pay for its inclusion in the
current block.  We assume that a user bids in order to maximize its
net gain (i.e., the value for inclusion minus the cost for inclusion).
To reason about the different possible bids for a transaction~$t$
submitted to a mempool~$M$, we use $M(\bidt)$ to denote
the result of adding the transaction~$t$ with bid $\bidt$ to~$M$.
For simplicity, we assume that each transaction in the current mempool
has a distinct creator.

\begin{definition}[User Utility Function]\label{def:userutil}
For a TFM $\tfm$, on-chain history~$B_1,$ $B_2,\ldots,B_{k-1}$, and
mempool~$M$, 
the utility of the originator of a transaction~$t \notin M$ with value
$\valt$ and bid $\bidt$ is
\begin{equation}\label{eq:userutil}
u_t(\bidt) :=
\left( \valt -
\underbrace{p_t(B_1,\ldots,B_{k-1},B_k)}_{\text{payment to miner (per-gas-unit)}} -
\underbrace{q_t(B_1,\ldots,B_{k-1},B_k)}_{\text{fee burn (per-gas-unit)}} \right)
\cdot g_t
\end{equation}
if~$t$ is included in $B_k = \allocs(B_1,\ldots,B_{k-1},M(\bidt))$
and~0 otherwise.
\end{definition}
In~\eqref{eq:userutil}, we highlight the dependence of the utility
function on the argument that is directly under a user's
control, the bid~$\bidt$ submitted with the transaction.
We assume that a transaction creator bids to maximize the utility
function in~\eqref{eq:userutil}.\footnote{While the creator of a
  transaction~$t$ has no direct control over $\allocs$, $\prices$, or
  $\burns$, its bid~$\bidt$ is embedded in~$M(\bidt)$ and therefore
  can affect~$B_k=\allocs(B_1,\ldots,B_{k-1},M(\bidt))$.  This, in turn,
can affect $\pricet(B_1,\ldots,B_{k-1},B_k)$ and
$\burnt(B_1,\ldots,B_{k-1},B_k)$.  For example, whether or not
$\alloct(B_1,\ldots,B_{k-1},M(\bidt))=1$ generally depends on whether or
not~$\bidt$ is large relative to the bids of competing transactions in~$M$.}

\subsubsection{Bidding Strategies and Ex Post Nash Equilibrium}

Intuitively, ``easy fee estimation'' should mean that the ``obvious''
bidding strategy is optimal.  Formally, a {\em bidding strategy} is a
function~$\epnebid$ that specifies a bid $\epnebid(\valt)$ for a
transaction~$t$ as a function of the value~$\valt$ of that
transaction.
A bidding strategy is a function of the value~$\valt$
only (which is known to the transaction creator) and not, for example,
bids submitted by competing transactions (which are not).\footnote{A
bidding strategy can depend also on the blockchain history (e.g., with
  EIP-1559, on the current base fee).  For the purposes of a
  single-block game-theoretic analysis, we can take the history as
  fixed and suppress this dependence in the notation.}
For example, a plausible
bidding strategy in a first-price auction is to shade one's bid, but
not by too much, perhaps by setting $\epnebid(\valt) = .75 
\valt$ for all $\valt$.

Suppose we have in mind an ``obvious'' bidding strategy
$\epnebid(\cdot)$ for users to employ.  What does it mean that bidding
in this obvious
way is ``always optimal''?  The answer is formalized by the concept of
a symmetric ex post Nash equilibrium (symmetric EPNE).  
Intuitively, obvious bidding
should maximize a user's utility as long as all the other users are
also bidding in the obvious way.\footnote{``Symmetric'' refers to the
  fact that the obvious bidding strategy~$\epnebid(\cdot)$ is the same
  for every transaction~$t$.}
\begin{definition}[Symmetric Ex Post Nash Equilibrium (Symmetric EPNE)]\label{def:epne}
Fix a TFM $\tfm$ and the on-chain history $B_1,B_2,\ldots,B_{k-1}$.
A bidding strategy $\epnebid(\cdot)$ is a {\em symmetric ex post Nash 
  equilibrium (symmetric EPNE)} if, for every mempool~$M$ in which
all transactions' bids were set according to this strategy, and for
every transaction~$t \notin M$ with
value~$v_t$, bidding~$\epnebid(v_t)$ maximizes the
utility~\eqref{eq:userutil} of~$t$'s creator.
\end{definition}
Crucially, following the bid recommendation~$\epnebid(\valt)$ of a
symmetric EPNE does not require reasoning about competing transactions
in~$M$, other than keeping the faith that their bids were set 
according to the bid recommendations of the symmetric
EPNE.\footnote{An even stronger notion is a {\em dominant-strategy
    equilibrium}, in which~$\epnebid(\valt)$ is optimal for~$t$'s creator
  no matter what the other users do.  ``Obvious bidding'' is not a
  dominant-strategy equilibrium in the transaction fee mechanism
  proposed in EIP-1559 (see Remark~\ref{rem:not_dse}), but it is in a
  variant with hard-coded tips
  (see Theorem~\ref{t:tipless_uic} and footnote~\ref{foot:dse2}).\label{foot:dse}}

We can now define a TFM to be incentive-compatible from the user
perspective if there's always an obvious bidding strategy in the form
of a symmetric EPNE.
\begin{definition}[Incentive-Compatibility for Users (UIC)]\label{def:uic}
A TFM $\tfm$ is {\em incentive-compatible for users (UIC)} if, for
every on-chain history $B_1,B_2,\ldots,B_{k-1}$, there is a symmetric
EPNE.
\end{definition}
In this report, we identify ``mechanisms with easy fee estimation''
and ``mechanisms with good UX'' with the UIC condition of Definition~\ref{def:uic}.

\begin{example}[First-Price Auctions Are Not UIC]\label{ex:fpa_uic}
First-price auctions are not easy to reason about, in the sense that
they are not UIC.  Intuitively, the utility-maximizing bid
depends on the precise numerical values of others' bids, and not merely
on the qualitative knowledge that they are following a particular
bidding strategy.

For example, consider a block with room for one transaction, a
transaction~$t$ with value~$\valt=10$, and
suppose that all transactions other than~$t$ use the same
bidding strategy $\epnebid(\val_s) = .75 \cdot v_s$.  
If the highest value of $\val_s$ of any transaction $s \neq t$ is 10,
then the highest bid by any such transaction will be 7.5, and
the utility-maximizing bid for~$t$'s creator will be~7.51.
If the highest other value is~8, the optimal bid is~6.01; and so on.
The key point is that the optimal bid to include with the transaction
is a function not only of that transaction's value, but also of the
values of  the competing transactions (even after assuming that all
their bids are set using a known bidding strategy $\epnebid(\cdot)$).
\end{example}
Thus, in a precise sense, first-price auctions do not offer ``good
UX'' in the form of an easy-to-follow optimal bid recommendation.
We'll see later (Theorem~\ref{t:1559uic}) that the
transaction fee mechanism proposed in EIP-1559 is UIC except during
periods of rapidly increasing demand.

\subsection{Off-Chain Agreements}\label{ss:oca}

The game-theoretic guarantees in Section~\ref{ss:mmic} concern 
attacks that manipulate the contents of a block (by including fake
transactions, or more generally deviating from the allocation intended
by the transaction fee mechanism).
This section treats a different type of attack that is also
implementable at the time scale of a single block, namely collusive
agreements between miners and users.  
Recall that a set~$T$ of transactions is {\em feasible} if the total
gas $\sum_{t \in T} g_t$ is at most the maximum block size~$G$.

\begin{definition}[Off-Chain Agreement (OCA)]\label{def:oca}
For a feasible set $T$ of transactions and a
miner~$m$, an {\em off-chain agreement   (OCA)} between 
$T$'s creators and~$m$ specifies:
\begin{itemize}

\item [(i)] a bid vector~$\bids$, with~$\bidt$ indicating the bid
to be submitted with the transaction~$t \in T$;

\item [(ii)] a per-gas-unit ETH transfer~$\tau_t$ from the creator of
  each transaction~$t \in T$ to the miner~$m$.

\end{itemize}
\end{definition}
In an OCA, 
each creator of a transaction~$t$ agrees to
submit~$t$ on-chain with a bid of $\bidt$ while transferring~$\tau_t$
per unit of gas to the miner~$m$ off-chain;
the miner, in turn, agrees to mine a block~$B(\bids)$ comprising the
transactions in~$T$ (with on-chain bids $\bids$). 

\begin{example}[Moving Payments Off-Chain]\label{ex:move_off_chain}
To get a feel for OCAs, imagine a first-price auction in which 50\% of
the revenue is burned and the other 50\% is transferred to the miner.
(See also Section~\ref{ss:fpa_burn}.)  Miners and users could then
collude as follows:
\begin{enumerate}

\item Users bid zero on-chain and communicate off-chain what
  they would have bid in a standard first-price auction.

\item Miners keep 75\% of the (off-chain) bids of the transactions
  they include,
  with the other 25\% refunded to those transactions' creators.

\end{enumerate}
In the notation of Definition~\ref{def:oca}, this is the OCA $\oca$ in
which $\bids = \mathbf{0}$ and $\tau_t = .75\bid'_t$, where $\bid'_t$
denotes what $t$'s creator would have bid in a first-price auction
without fee-burning.  Compared to the ``honest'' on-chain outcome with
bids $\bids'$, miners earn 50\% more revenue and users enjoy a 25\%
discount, both at the expense of the network.
\end{example}

Given a TFM $\tfm$ and on-chain history~$B_1,B_2,\ldots,B_{k-1}$,
the utility of $t$'s creator from such an OCA $\oca$ is
given by the right-hand side of~\eqref{eq:userutil}, less its transfer to
the miner:
\begin{equation}\label{eq:oca_user}
\left( \valt - p_t(B_1,\ldots,B_{k-1},B(\bids)) -
  q_t(B_1,\ldots,B_{k-1},B(\bids)) - \tau_t \right) \cdot g_t.
\end{equation}
(Users not part of~$T$ receive zero utility.)
The miner's utility is given by the sum of on-chain and off-chain
payments received, less the costs incurred:
\begin{equation}\label{eq:oca_miner}
\sum_{t \in T} \left( \pricet(B_1,B_2,\ldots,B_{k-1},B(\bids))
+ \tau_t - \mu \right) \cdot  g_t.
\end{equation}
Adding up these utility functions---one per
transaction~$t \in T$, plus one for the miner---results in the joint
utility enjoyed by all parties in an OCA~$\oca$:
\[
u_{T,m}\oca := \sum_{t \in T} \left( \valt -
  q_t(B_1,\ldots,B_{k-1},B(\bids)) - \mu \right)
\cdot g_t.
\]
From the coalition's perspective, on-chain and off-chain payments from
the users to the miner (the $\pricet$'s and $\tau_t$'s) remain
within the coalition and thus cancel out; the fee burn (the
$\burnt$'s) is transferred outside the coalition (to the network) and
is therefore a loss.  Thus, the point of an OCA is to maximize the
joint utility---the amount of transaction value that is not lost to
the protocol or to the miner's costs.
\begin{definition}[Joint Utility]
For an on-chain history~$B_1,B_2,\ldots,B_{k-1}$,
the {\em joint utility} of the miner and users for the block~$B_k$ is
\begin{equation}\label{eq:jointutil}
\sum_{t \in B_k} \left( \valt - q_t(B_1,B_2,\ldots,B_{k-1},B_k) - \mu
\right) \cdot g_t.
\end{equation}
\end{definition}
We assume that miners and users act to maximize their joint utility.
Using transfers, a miner and users 
can then split this joint utility among themselves in an arbitrary
way.\footnote{For example, suppose an OCA increases the joint utility
  of a coalition by increasing the utility of six
  users by~1 ETH each while decreasing the miner's utility
  by~5 ETH.  The OCA transfers can then be adjusted so that all parties
enjoy strictly higher individual utility, for example by sending an
extra  $\tfrac{11}{12}$ ETH from each of these users to the miner.
Additional transfers can be used to also strictly increase the utility
of the creators of the transactions 
excluded from the block~$B_k$.}
For this reason, when analyzing OCAs, we can focus on the joint
utility~\eqref{eq:jointutil} of the miner and the creators of the
included transactions, without concern about how it might be split
among them and
the creators of the excluded transactions.

A TFM is then {\em OCA-proof} if, for every OCA,
there is an equally good on-chain outcome. 
For a set of transactions~$U$ and bids~$\bids$ for those transactions,
we denote by~$U(\bids)$ the corresponding mempool.
\begin{definition}[OCA-Proof]\label{def:ocaproof}
A TFM $\tfm$ is {\em OCA-proof} if, for every on-chain
history~$B_1,$ $B_2,\ldots,B_{k-1}$ and set~$U$ of outstanding
transactions, there exists bids~$\bids^*$ for the transactions of~$U$
such that, for the resulting on-chain outcome $B_k = \allocs(B_1,
B_2,\ldots,B_{k-1},U(\bids^*))$, 
\begin{equation}\label{eq:oca}
\underbrace{\sum_{t \in B_k} \left( \valt -
  q_t(B_1,\ldots,B_{k-1},B_k) - \mu \right)
\cdot g_t}_{\text{joint utility of on-chain outcome}}
\ge u_{T,m}\oca
\end{equation}
for every feasible subset~$T \sse U$ of transactions and OCA $\oca$
between their creators and the miner~$m$.
\end{definition}
In other words, if a TFM is {\em not} OCA-proof, there are scenarios in
which a miner and users can collude to achieve higher joint
utility---and, after defining appropriate transfers, higher
individual utilities---than in any on-chain outcome.

Intuitively, first-price auctions are OCA-proof because off-chain
payments can be costlessly replaced by on-chain bids.  The next
example formally verifies Definition~\ref{def:ocaproof}.
\begin{example}[First-Price Auctions Are OCA-Proof]
Consider a set~$U$ of transactions and set $\bid^*_t = \valt$ for every
$t \in U$.
Then, because~$\burns^f$ is the all-zero function
(Example~\ref{ex:fpa_burn}), the objective~\eqref{eq:fpa_obj}
maximized by the allocation rule~$\allocs^f$ is identical to the joint
utility~\eqref{eq:jointutil}.  Thus, the joint utility of the on-chain
outcome with bids $\bids^*$ cannot be improved upon by any OCA.
\end{example}

\begin{remark}[OCA-Proofness and Fee Burning]
OCAs are the biggest game-theoretic driver for the why and the
how of the fee burn in the transaction fee mechanism proposed in
EIP-1559.   For example, adding a fee burn to a first-price auction
destroys its OCA-proofness (Section~\ref{ss:fpa_burn}).  Meanwhile,
because of OCAs, a history-dependent base fee has no teeth unless
revenue from it is burned or otherwise withheld from the miner (Section~\ref{ss:refund}).
\end{remark}

\section{Formal Analysis of the 1559 Mechanism with Myopic Miners}\label{s:mm}

This section investigates to what extent the transaction fee mechanism
proposed in EIP-1559---henceforth, the {\em 1559
  mechanism}---satisfies the three game-theoretic guarantees identified in
Section~\ref{s:defs} (MMIC, UIC, and OCA-proofness).
Section~\ref{ss:1559desc} translates the description of the mechanism
in Section~\ref{ss:1559} into the formalism introduced in
Section~\ref{s:defs}.  Sections~\ref{ss:1559mmic}--\ref{ss:1559oca}
prove that the mechanism is always MMIC and OCA-proof, and is UIC
except during periods of rapidly increasing demand.

\begin{mdframed}[style=offset,frametitle={Game-Theoretic Guarantees
    for the 1559 Mechanism}]
\begin{enumerate}

\item Myopic miners are incentivized to follow the intended allocation
  rule, and are strictly disincentivized from including fake
  transactions in a block.

\item Except in periods of a large and sudden demand spike, there are
``obvious'' optimal bids for users: set a transaction's fee cap to its
value and its tip to cover the marginal cost of gas to the miner.

\item Miners and users can never improve their joint utility through
  an off-chain agreement.

\end{enumerate}
\end{mdframed}

\subsection{The 1559 Mechanism}\label{ss:1559desc}

Recall from Section~\ref{ss:1559} that, in the 1559 mechanism,
each block is associated with a base fee that is fixed by
the history of past blocks and independent of
the contents of the current block; we denote by
$\alpha(B_1,B_2,\ldots,B_{k-1})$ the base fee for the next block that
is determined by a
particular history~$B_1,B_2,\ldots,B_{k-1}$.  The specific
function~$\alpha$ proposed in EIP-1559 is the iteration of the base
fee update rule in~\eqref{eq:update}, although these details will not
be important for the single-block game-theoretic analysis carried out in
this section.

Recall also that, in EIP-1559, each transaction specifies a tip
$\delta_t$ and a fee cap~$c_t$.  These two parameters induce a bid
$\bidt$ for the transaction with respect to any given base fee~$r$,
namely 
\begin{equation}\label{eq:1559bid}
\bidt = \min\{ r + \delta_t, c_t \}.
\end{equation}
\begin{definition}[1559 Allocation Rule]\label{def:1559alloc}
For each history~$B_1,B_2,\ldots,B_{k-1}$ and corresponding base fee
$r = \alpha(B_1,B_2,\ldots,B_{k-1})$, 
the (intended) allocation rule~$\allocs^{\nine}$ of the 1559 mechanism 
is to include a feasible subset of outstanding transactions that
maximizes the sum of the gas-weighted bids,
less the gas costs and total base fee paid.
That is, the $\alloct^{\nine}$'s are assigned 0-1 values to maximize
\begin{equation}\label{eq:1559_obj}
\sum_{t \in M} \alloct^{\nine}(B_1,B_2,\ldots,B_{k-1},M) \cdot
(\bidt - r - \mu)
  \cdot g_t,
\end{equation}
subject to the block size constraint~\eqref{eq:feasible}.
\end{definition}
The payment rule transfers the difference between the bid and the base
fee to the miner.
\begin{definition}[1559 Payment Rule]\label{def:1559payment}
In the 1559 mechanism, letting~$r = \alpha(B_1,B_2,\ldots,B_{k-1})$,
\[
\pricet^{\nine}(B_1,B_2,\ldots,B_{k-1},B_k) = b_t - r
\]
for all~$B_1,B_2,\ldots,B_k$ and~$t \in B_k$.
\end{definition}
The burning rule burns the base fee.
\begin{definition}[1559 Burning Rule]\label{def:1559burning}
In the 1559 mechanism, letting~$r = \alpha(B_1,B_2,\ldots,B_{k-1})$,
\[
\burnt^{\nine}(B_1,B_2,\ldots,B_{k-1},B_k) = r
\]
for all~$B_1,B_2,\ldots,B_k$ and~$t \in B_k$.
\end{definition}
Formally, the {\em 1559 mechanism} is the TFM mathematically encoded
by the triple of rules $(\allocs^{\nine},\prices^{\nine},\burns^\nine)$
described in Definitions~\ref{def:1559alloc}--\ref{def:1559burning}.

\subsection{The 1559 Mechanism Is Incentive Compatible for Myopic
  Miners}\label{ss:1559mmic}

This section evaluates the 1559 mechanism from the perspective of
myopic miners, and specifically the MMIC property
(Definition~\ref{def:mmic}) and $\cost$-costliness
(Definition~\ref{def:costly}).

\begin{theorem}[The 1559 Mechanism is MMIC]\label{t:1559mmic}
The 1559 mechanism $\ninetfm$ is MMIC.
\end{theorem}

\begin{proof}
Fix an on-chain history~$B_1,B_2,\ldots,B_{k-1}$, a mempool~$M$, and a
marginal cost of gas~$\mu \ge 0$ (as in Remark~\ref{rem:marg_costs}).
Let~$r$ denote the corresponding base fee
$\alpha(B_1,B_2,\ldots,B_{k-1})$ for the current block.
Substituting in Definitions~\ref{def:1559payment} and~\ref{def:1559burning},
myopic miner utility~\eqref{eq:mmutil} equals
\begin{equation}\label{eq:1559mmic}
u(F,B_k) = \underbrace{\sum_{t \in B_k \cap M} (\bidt-r-\mu) \cdot
  g_t}_{\text{net revenue from~$B_k$}}
-
\underbrace{\sum_{t \in B_k \cap F} (r+\mu) \cdot g_t}_{\text{cost of
    fake txs}},
\end{equation}
where~$B_k$ denotes the transactions included by the miner and~$F$ the
fake transactions that it creates.  Included fake transactions strictly
increase the second term (by~$r+\mu$ per unit of gas) while leaving the
first unaffected, so a myopic miner will only include real
transactions in~$B_k$.  In this case, myopic miner utility equals
\[
\sum_{t \in B_k} (\bidt-r-\mu) \cdot g_t,
\]
which is identical to the quantity~\eqref{eq:1559_obj} maximized by
the allocation rule $\allocs^{\nine}$ (Definition~\ref{def:1559alloc}).
Thus, myopic miner utility is maximized by following the allocation
rule and setting $B_k$ equal to $\allocs^{\nine}(B_1,B_2,\ldots,B_{k-1},M)$.
\end{proof}

From the expression~\eqref{eq:1559mmic} for myopic miner utility in
the 1559 mechanism, 
we can see immediately that it is
$\cost$-costly (Definition~\ref{def:costly}) for $\cost = r+\mu$.
\begin{corollary}[The 1559 Mechanism is $(r+\mu)$-Costly]\label{cor:1559costly}
Fix an on-chain history $B_1,B_2,\ldots,B_{k-1}$ and corresponding
base fee $r=\alpha(B_1,B_2,\ldots,B_{k-1})$ for the current block,
a mempool~$M$, and a marginal cost of gas~$\mu \ge 0$. 
The 1559 mechanism is $(r+\mu)$-costly.
\end{corollary}
\begin{remark}[Role of the Fee Burn]\label{rem:feeburn}
If the base fee was paid to miners rather than burned, the 1559
mechanism would only be $\mu$-costly and fake transactions would be
only mildly disincentivized.  The primary motivation for the fee
burn, however, is to rule out its evasion by off-chain agreements
(see Section~\ref{ss:refund}).
\end{remark}

\subsection{The 1559 Mechanism Is Typically Incentive Compatible for Users}\label{ss:1559uic}

The 1559 mechanism is always incentive compatible for myopic miners,
no matter what the current base fee and demand for block space
(Theorem~\ref{t:1559mmic}).  We next show that
the mechanism is also incentive compatible for users, except in
periods of rapidly increasing demand.

\subsubsection{Excessively Low Base Fees}

The next definition is a proxy for a period of rapidly increasing
demand.
\begin{definition}[Excessively Low Base Fee]\label{def:low}
Let~$\mu$ denote the marginal cost per unit of gas.
A base fee~$r$ is {\em excessively low} for a mempool~$M$ of transactions
if the demand at price~$r+\mu$ exceeds the maximum block size~$G$:
\begin{equation}\label{eq:low}
\underbrace{\sum_{t \in M \,:\, \valt \ge r+\mu} g_t}_{\text{demand at
    price $r+\mu$}} > G.
\end{equation}
\end{definition}
Excessively low base fees arise from large and sudden demand
spikes.  In Example~\ref{ex:traj}
in Section~\ref{ss:toohigh}, for instance, none of the eight periods
suffer from an excessively low base fee, despite the sudden doubling
of demand.  Modifying that example so that demand more than doubles in
period~2, there is a sequence of periods with excessively low
base fees,  ending once the base fee has increased enough to bring demand
back down below 25M gas (Table~\ref{table:basefee2}).

\begin{table}[H]
\centering
{\footnotesize
\begin{tabular}{|c|c|c|c|c|c|c|c|c|}\hline
& Period 1 & Period 2 & Period 3 & Period 4 & Period 5 & Period 6 & Period 7  & Period 8\\ \hline\hline
Demand & Low & High & High & High & High & High & High & Low \\ \hline
\hline
EIP-1559 Base Fee 
& $33.33$ & $33.33$ & $37.5$ & $42.18$ & $47.46$ & $53.39$ & $60.06$ & $66.19$
\\ \hline
EIP-1559 Block Size 
& 12.5M & 25M & 25M & 25M & 25M & 25M & 24.49M & 10.04M
\\ \hline
Excessively low? & No & Yes & Yes & Yes & Yes & Yes & No  & No\\ \hline
\end{tabular}
}
\caption{An example of excessively low base fees due to a large and
  sudden jump in demand.  The marginal cost~$\mu$ of gas is~0.
``Low'' demand means the demand curve
  $D(p)=15000000-75000p$; ``high''
  means the demand curve $D(p)=35000000-175000p$.
(Here ``demand'' means the total gas
  consumed by all pending transactions with a value of~$p$ or
  more.)
The second and third rows show
the joint evolution of the base fee and block size under the EIP-1559
mechanism, assuming that the base fee matches the market-clearing
price in period~1 and that all users submit a bid equal to the minimum
of their value and the base fee.
Periods~2--6 suffer from excessively low base fees.
}\label{table:basefee2}
\end{table}

\subsubsection{The 1559 Mechanism Is UIC Except with Excessively Low
  Base Fees}\label{sss:1559uic}

When the base fee is excessively low, users must compete for scarce
block space through their tips, and the 1559 mechanism effectively
reverts back to a first-price auction.  As first-price auctions are
essentially never UIC (see Example~\ref{ex:fpa_uic}), the 1559
mechanism is not UIC when
the base fee is excessively low.  The good news is that an excessively
low base fee is the only reason why the 1559 mechanism might fail to
be UIC.  That is, whenever the base fee is not excessively low, there
is an ``obvious optimal bid'' in the form of a symmetric EPNE
(Definition~\ref{def:epne}).  This optimal bid corresponds to setting
a transaction's fee cap equal to its creator's value (i.e., $c_t =
\valt$), and a
transaction's tip equal to the marginal cost of gas to a miner (i.e.,
$\delta_t = \mu$).  
\begin{theorem}[The 1559 Mechanism Is Typically UIC]\label{t:1559uic}
Fix an on-chain history $B_1,B_2,\ldots,B_{k-1}$ and corresponding
base fee $r = \alpha(B_1,B_2,\ldots,B_{k-1})$, a marginal cost~$\mu$
of gas to miners, and a mempool~$M$ of transactions for which~$r$ is not
excessively low.  The bidding strategy
\begin{equation}\label{eq:epne}
\epnebid(\valt) = \min\{ r+ \mu, \valt \}
\end{equation}
constitutes a symmetric EPNE under the 1559 mechanism.
\end{theorem}

\begin{proof}
Suppose each creator of a transaction $t \in M$ sets its bid
according to the strategy~$\epnebid(\cdot)$ in~\eqref{eq:epne}; we
need to show that no creator could increase its expected
utility~\eqref{eq:userutil} by changing its bid (holding the bids of
other transactions fixed).  

The objective~\eqref{eq:1559_obj} of the 1559 allocation rule
prescribes including precisely the transactions~$t \in M$ with $\bidt
\ge r + \mu$.  
Because $\epnebid(\valt) = \min\{ r+ \mu, \valt \}$ for all $t \in M$, 
these are precisely the transactions~$t \in M$ with $\valt \ge r +
\mu$.  In particular, because~$r$ is not excessively low for~$M$,
this allocation is feasible:
\begin{equation}\label{eq:1559uic}
\underbrace{\sum_{t \in M \,:\, \epnebid(\valt) \ge r + \mu}
  g_t}_{\text{gas of included txs}}
\quad = 
\underbrace{\sum_{t \in M \,:\, \valt \ge r + \mu}
  g_t}_{\text{demand at price $r + \mu$}} \le\quad G.
\end{equation}

There are two types of transactions~$t$ to consider, high-value ($\valt
\ge r + \mu$) and low-value ($\valt < r + \mu$); see also
Table~\ref{table:uic}. 
When all bids are set according the strategy~$\epnebid(\cdot)$ in~\eqref{eq:epne},
the former
transactions are included (and pay $\epnebid(\valt) = r + \mu$ per
unit of gas) while the latter are excluded (and pay nothing).  
The utility~\eqref{eq:userutil} of $t$'s creator is $(\valt - r- \mu)
\cdot g_t \ge 0$ if~$t$ is a high-value transaction and~0 otherwise.
Every alternative bid $\hat{\bid}_t$ for a high-value transaction
either has no effect on its creator's utility (if $\hat{\bid}_t \ge
r+\mu$) or leads to~$t$'s exclusion from the block (if $\hat{\bid}_t <
r + \mu$) and reduces this utility from $(\valt - r- \mu) \cdot g_t$
to~0.
Every alternative 
bid $\hat{\bid}_t$ for a low-value transaction
either has no effect on its creator's utility or leads to~$t$'s
inclusion in the block; the latter can only occur when $\hat{\bid}_t \ge
r+\mu$, in which case the creator's utility drops from~0 to $(\valt
-\hat{\bid}_t)
\cdot g_t < 0$.  We conclude that there is no alternative bid for
any transaction of~$M$ that increases its creator's utility.
\end{proof}

Theorem~\ref{t:1559uic} and its proof show that, at its symmetric
EPNE, the 1559 mechanism acts a posted-price mechanism
(Section~\ref{ss:postedprice}) except when the base fee is excessively
low. 
\begin{mdframed}[style=offset,frametitle={The 1559 Mechanism Is
    Typically a Posted-Price Mechanism}]
The 1559 mechanism acts as a posted-price mechanism at the price
$r+\mu$, where~$r$ is the base fee and~$\mu$ is the marginal cost of
gas, except during periods of rapidly increasing demand.
\end{mdframed}

\begin{table}
\centering
\begin{tabular}{|c|c|c|}\hline
& Low-Value ($\valt < r + \mu$) & High-Value ($\valt \ge r + \mu$)\\ \hline
Bid at EPNE & $\valt$ & $r+\mu$\\
Utility at EPNE & 0 & $(\valt - r - \mu) \cdot g_t \ge 0$\\
Utility of Alternative & $\le (\valt - r - \mu) \cdot g_t < 0$ & 0\\ \hline
\end{tabular}
\caption{Proof of Theorem~\ref{t:1559uic}.  For both low- and high-value
  transactions, no unilateral deviation from the symmetric EPNE bid 
can increase a user's utility.}\label{table:uic}
\end{table}

\begin{remark}[Welfare Properties of the 1559 Mechanism]
An attractive property of the symmetric EPNE in~\eqref{eq:epne} is
that the outcome perfectly differentiates between high-value ($\valt
\ge r + \mu$) and low-value ($\valt < r + \mu$) transactions,
including the former while excluding the latter.  This outcome can be
viewed as a market-clearing outcome (Section~\ref{ss:market}) with
respect to a supply of~$G^*$ gas, where~$G^*$ denotes the demand at
price $r+\mu$.
\end{remark}

\begin{remark}[The Obvious Bid Is Not a Dominant
  Strategy]\label{rem:not_dse}
The symmetric EPNE~\eqref{eq:epne} in the proof of Theorem~\ref{t:1559uic} is not a
dominant-strategy equilibrium in the sense of
footnote~\ref{foot:dse}.  The issue arises when the creators of other 
transactions overstate their fee caps, in which case the base fee
could become excessively low with respect to the stated demand (even
though it is not with respect to the true demand).  In particular,
the equality in~\eqref{eq:1559uic} need not hold if other
transactions' bids are set arbitrarily.
\end{remark}

\begin{remark}[Expected Frequency of Excessively Low Base Fees]
Demand for EVM computation has generally been volatile, at both
short and long time scales.  For this reason, one would expect
at least occasional excessively low base fees.  It would be
interesting to predict, perhaps based on experiments using historical
demand data, the likely frequency of excessively low base fees in a post-EIP-1559 world.
\end{remark}

\subsection{The 1559 Mechanism Is OCA-Proof}\label{ss:1559oca}

Finally, we show that, under the 1559 mechanism, miners and users
cannot improve their joint utility through off-chain agreements.
A key driver of this result is that the fee burn (per unit of
gas) does not depend on the current actions of the miner or users
(cf., Section~\ref{ss:fpa_burn}).

\begin{theorem}[The 1559 Mechanism is OCA-Proof]\label{t:1559ocaproof}
The 1559 mechanism $\ninetfm$ is OCA-proof.
\end{theorem}

\begin{proof}
Fix an on-chain history~$B_1,B_2,\ldots,B_{k-1}$ and corresponding
base fee $r=\alpha(B_1,B_2,\ldots,B_{k-1})$.
Consider a set~$U$ of transactions and set $\bid^*_t = \valt$ for every
$t \in U$.  
Then, because~$\burns^{\nine}$ is the constant function always equal to~$r$ (Definition~\ref{def:1559burning}),
the objective~\eqref{eq:1559_obj}
maximized by the allocation rule~$\allocs^{\nine}$ is identical to the joint
utility~\eqref{eq:jointutil}.  Thus, the joint utility of the on-chain
outcome with bids $\bids^*$ cannot be improved upon by any OCA.
\end{proof}

\section{Miner Collusion at Longer Time Scales}\label{s:collusion}

Section~\ref{s:mm} demonstrates that the 1559 mechanism enjoys several
game-theoretic guarantees at the time scale of a single block.  But
what about longer time scales?  For example, 
to achieve the typically-UIC guarantee in Theorem~\ref{t:1559uic}, 
the mechanism introduces a history-dependent base fee that is burned;
a natural worry is that miners may be incentivized to manipulate and
artificially decrease this base fee over time.

This section investigates the incentives for miner collusion, both under
the status quo and under EIP-1559.
Section~\ref{ss:thought} formalizes ``extreme miner collusion''
through a thought experiment in which a single miner controls 100\% of
Ethereum's hashrate.  Section~\ref{ss:fpa_thought} 
identifies the revenue-maximizing strategy for such a miner in a first-price
auction; in some cases, the miner is incentivized to 
artificially restrict the supply of EVM computation in order to boost
the bids submitted by creators of high-value transactions.
Section~\ref{ss:1559_thought} repeats the exercise for the 1559
mechanism and determines that the outcome of extreme collusion would
be similar to that with today's first-price auctions.
Section~\ref{ss:fpa_collusion} classifies different types of miner
collusion and reviews to what extent each type appears to occur
in Ethereum at present.  
Section~\ref{ss:1559collusion} argues that the game-theoretic impediments 
to double-spend, censorship, denial-of-service, and
revenue-maximizing 100\% miner strategies (including base fee
manipulation) appear as strong under EIP-1559 as under the status quo.
Finally, Section~\ref{ss:caveats} brainstorms possible reasons for why
miner collusion might nevertheless be more likely under EIP-1559 than
it is today.

\subsection{Extreme Collusion: The 100\% Miner Thought Experiment}\label{ss:thought}

The fidelity of the myopic miner model of
Sections~\ref{s:defs}--\ref{s:mm} depends on the degree of
decentralization in Ethereum mining.  For example, with extreme
decentralization, such as the hashrate being spread equally across
millions of non-colluding miners, any given miner mines a block so
rarely that there is no point to non-myopic strategies (i.e.,
strategies that forego immediate rewards in favor of future rewards).
In particular, in the 1559 mechanism, because the base fee is set by
past history and independent of the current block, 
no such miner will be interested in manipulating it.  

To meaningfully study miner deviations such as base fee manipulation,
we must therefore consider miners (or tightly coordinated mining
pools) that possess a significant fraction of the total hashrate and
strategize at time scales longer than a single block.\footnote{We
  continue to assume that users are myopic, and bid to maximize their
  utility in the current block (Definition~\ref{def:userutil}).
  Simulations by Monnot~\cite{1559sim} suggest that more complex user
  strategies do not significantly change the behavior of the mechanism
  proposed in EIP-1559.}  To get the lay of the land, we next
investigate both first-price auctions and the 1559 mechanism in the
opposite extreme scenario in which {\em all} of the hashrate is
controlled by a single miner or, equivalently, a perfectly coordinated
cartel comprising all of the miners.\footnote{A similar approach is
  taken by Hasu et al.~\cite{HPC19} in the context of Bitcoin and
  Zoltu~\cite{MZ51} in the context of EIP-1559.}
\begin{mdframed}[style=offset,frametitle={The 100\% Miner Thought
    Experiment}]
\begin{enumerate}

\item A single miner controls 100\% of the hashrate.

\item The miner acts to maximize its net revenue received from
  transaction fees over a significant period of time (e.g., thousands of
  blocks).

\item The demand curve (see Section~\ref{s:market}) is the same for
  every block, independent of the miner's actions, and known to the
  miner.

\end{enumerate}
\end{mdframed}
The second assumption clarifies that the 
thought experiments in Sections~\ref{ss:fpa_thought}
and~\ref{ss:1559_thought} will
not consider off-chain rewards, for example from a double-spend
attack, in order to isolate incentive issues specific to the
transaction fee mechanism.
The point of the third assumption is to stack the deck against a
protocol by making it as easy as possible for a miner or cartel of
miners to identify and carry out optimal deviations from the
protocol's prescriptions.

\subsection{First-Price Auctions with a 100\% Miner}\label{ss:fpa_thought}

What would a 100\% miner do under the status quo of first-price
auctions?  Let~$D(p)$ denote the demand curve---the total gas demanded
at a gas price of~$p$ gwei.  We assume that~$D(p)$ is a continuous and
strictly decreasing function, and that~$D(p) = 0$ once~$p$ is
sufficiently large.  We continue to assume that the 
demand curve is exogenous, the same for every block, and known to the
miner.

We consider strategies of the following form:
\begin{mdframed}[style=offset,frametitle={Strategies for a 100\% Miner}]
\begin{enumerate}

\item \textbf{Price-setting: } for a gas price~$p$ with $D(p) \le G$,
  include a transaction in the block if and only if its gas price is
  at least~$p$.  (As usual, $G$ denotes the maximum block
  size.)\tablefootnote{For an analogy, think of a consultant with a
    unique skill set committing to an hourly rate.}

\item \textbf{Quantity-setting: } for a quantity $q \le G$,
  include the transactions with the highest gas prices, up to a limit
  of~$q$ on the total gas.\tablefootnote{For an analogy, think of
a restriction on oil production set by the Organization of the Petroleum
Exporting Countries (OPEC).}

\end{enumerate}
\end{mdframed}
In our model, these two types of strategies are equivalent---a price-setting strategy at the price~$p$ has the same
effect as a quantity-setting strategy at the quantity~$q=D(p)$.
In either case, a creator of a transaction~$t$ 
with $\valt \ge p$ should be expected to respond by bidding the fixed
price~$p$ (enough for inclusion in the block),
and one with $\valt < p$ to bid something between~0
and~$\valt$ (in any case, being excluded from the block).

Because there are no dependencies between first-price auctions in
different blocks and no fee burn, a 100\% miner maximizes its net
revenue by maximizing its revenue from each block separately.  For a
single block, the revenue earned by a miner using a price-setting strategy
with price~$p$ (or the equivalent quantity-setting strategy) is the
price 
times the quantity willing to pay it:\footnote{For simplicity, we
  assume in this section that the marginal cost of gas to a
  miner---the parameter~$\mu$
  in Sections~\ref{s:defs}--\ref{s:mm}---is zero.  The conclusions
  of this section remain the same for a positive marginal cost.}
\begin{equation}\label{eq:pq}
R(p) := p \cdot D(p).
\end{equation}
For ease of exposition, in this section we focus on demand curves for
which the revenue~\eqref{eq:pq} is a strictly concave function (as is
the case with, for example, a linear demand curve).

Because a 100\% miner can be thought of as a monopoly on EVM
computation, there is an obvious price and quantity to focus on:
\begin{definition}[Monopoly Price and Quantity]\label{def:monopoly}
Consider a maximum block size~$G$ and a demand curve~$D(\cdot)$ for which
the revenue~\eqref{eq:pq} is a strictly concave function of price.
If~$\bar{p}$ attains the maximum in~\eqref{eq:pq}, then:
\begin{itemize}

\item [(a)] the {\em monopoly price} is the revenue-maximizing price
  or the market-clearing price, whichever is larger:\tablefootnote{See
    Section~\ref{ss:market} for the definition of a market-clearing
    price.}
\begin{equation}
p^* := \max\{ \bar{p}, D^{-1}(G) \};
\end{equation}

\item [(b)] the {\em monopoly quantity} is the revenue-maximizing 
quantity or the maximum block size, whichever is smaller:
\begin{equation}
q^* := D(p^*) = \min\{ D(\bar{p}), G \}.
\end{equation}

\end{itemize}
\end{definition}
\begin{example}[Monopoly Prices and Quantites]
Suppose the maximum block size~$G$ is 12.5M gas and the demand curve
is~$D(p) = 30000000-150000p$ (as in Figure~\ref{f:demand_curve}).  The
revenue~\eqref{eq:pq} as a function of price is
$30000000p-150000p^2$.  This function is differentiable and strictly
concave, so its unique maximum~$\bar{p}$ is the point at which the
derivative $30000000-300000p$ equals 0.  Thus, $\bar{p} = 100$ gwei and
$D(\bar{p}) = \text{15M}$ gas.  This exceeds the maximum block size of
12.5M gas, and hence the monopoly price is the market-clearing price
$D^{-1}(\text{12.5M}) = 116 \tfrac{2}{3}$ gwei.

If instead the demand curve was~$D(p) = 20000000-150000p$, $\bar{p}$
would be~$66 \tfrac{2}{3}$ gwei and~$D(\bar{p})$ would be~10M gas.
In this case, the monopoly price is strictly higher
than the market-clearing price (of~50 gwei) and the monopoly quantity is
strictly smaller than the maximum block size.
\end{example}

Price-setting at the monopoly price or quantity-setting at the
monopoly quantity both have the effect of maximizing
revenue~\eqref{eq:pq} subject to the maximum block size.  That is,
these are precisely the optimal strategies for a 100\% miner:
\begin{mdframed}[style=offset,frametitle={Optimal Strategies for a
    100\% Miner (First-Price Auctions)}]
\begin{itemize}

\item A 100\% miner would price-set at the monopoly price or
  quantity-set at the monopoly quantity.  

\item If the monopoly quantity is the maximum gas size (equivalently, the
monopoly price is the market-clearing price), a
  100\% miner would not deviate from the intended allocation rule of a 
  first-price auction (Example~\ref{ex:fpa_alloc}).

\end{itemize}
\end{mdframed}
We conclude that, with a first-price auction, extreme miner collusion
can increase transaction fee revenue if and only if the monopoly
quantity is less than the maximum block size; in this case, 
miners can boost revenues by artificially restricting the supply of
EVM computation, thereby forcing the creators of high-value
transactions to submit higher bids for their inclusion.

\begin{remark}[Detecting a Price- or Quantity-Setting Strategy]\label{rem:monopoly_detect}
Suppose a cartel of miners implemented a quantity-setting
strategy (with quantity less than~12.5M gas), or the corresponding
price-setting strategy.  Would anyone notice?
Naively executed, persistent underfull blocks would be a dead
giveaway.  But if miners include fake transactions to keep all the
blocks full, such a strategy could be difficult to conclusively
detect.
\end{remark}

\subsection{EIP-1559 with a 100\% Miner}\label{ss:1559_thought}

EIP-1559, described in Section~\ref{ss:1559}, would have two immediate
consequences for the 100\% miner thought experiment.  First, a miner
would strive to simultaneously maximize the transaction fee revenue (as
in a first-price auction) and minimize the amount of these fees lost
to the fee burn.  Second, a miner's allocation decision in one block
would affect the base fee (and hence net revenue earned) in future
blocks.  Now what's the miner's optimal strategy?

First suppose that the monopoly quantity
(Definition~\ref{def:monopoly}) is at most the target block size
of~12.5M gas.  (Recall that the maximum block size~$G$ is double this
amount under EIP-1559.)  In this case, the maximum block size may as
well be~12.5M gas---a 100\% miner will never use more than this, as
doing so would decrease its net revenue both in the current block (for
including more gas than the monopoly quantity) and in future blocks
(because of the increased base fee, as per~\eqref{eq:update}).  Thus,
the best-case scenario for a 100\% miner is to match the revenue of a
100\% miner under the status quo (Section~\ref{ss:fpa_thought}) while
simultaneously paying no fee burn.  A 100\% miner can closely
approximate this best-case scenario:
\begin{mdframed}[style=offset,frametitle={Optimal Strategy for a 100\%
    Miner (Monopoly Quantity $\le \text{12.5M gas}$)}]
\begin{enumerate}

\item Drive the base fee to
zero (or its minimum amount) from its initial value, for example
by publishing a sequence of empty blocks.

\item For all future blocks, proceed as a 100\% miner would with
  first-price auctions, by using the monopoly quantity-setting strategy.

\end{enumerate}
\end{mdframed}
Because the monopoly quantity is at most the target block size, the
base fee will remain at its minimum value forevermore.  
Notably, in this case, the outcome of extreme miner
collusion is essentially the same under EIP-1559 as under the status quo!

When the monopoly quantity is more than the target block size, a
100\% miner faces the non-trivial optimization problem of optimally
trading off the short-term revenue gain from including more than
12.5M gas in a block and the long-term revenue decrease due to higher
future base fees.  The optimal solution to this problem depends in an
intricate way on the assumed demand curve~$D(\cdot)$; a detailed
discussion of it is outside the scope of this report.  
Qualitatively, we can view this optimal strategy as an optimized
version of the quantity-setting strategy with quantity 12.5M gas,
in which the variable block size of EIP-1559 is exploited to mix
underfull and overfull blocks so as to boost net revenue (even
after accounting for the nonzero fee burn).

\begin{remark}[51\% Miner = 100\% Miner]
A miner or perfectly coordinated cartel of miners controlling 51\% of the
overall hashrate can control 100\% of the blocks on the longest chain
by refusing to extend any block mined by a miner outside the cartel.
Thus, the optimal 100\% miner strategies identified in this and the
previous section are equally well available to a 51\%
miner\footnote{Or at least, to a 51\% miner unconcerned with detectable
coordinated strategies; see Section~\ref{sss:detectable} for further
discussion.}\fnsep\footnote{Even a medium-size miner or mining
pool---over 20\% of the total hashrate, say, as with the two biggest
Ethereum mining pools~\cite{miningpools}---could conceivably benefit
from a monopoly price-setting or quantity-setting strategy, if enough
users are willing to pay a premium to avoid a 20\% chance of a
transaction being delayed by one block.\label{foot:mediumsize}}
\end{remark}

\subsection{First-Price Auctions: Do Miners Collude?}\label{ss:fpa_collusion}

We offer no prediction on whether miners would collude under EIP-1559,
for example to implement some form of the 100\% miner optimal strategy
identified in Section~\ref{ss:1559_thought}.  We can, however,
speculate in a principled way via analogy with observed miner
behavior under the status quo.

\subsubsection{Types of Miner Coordination}

Miners can coordinate their actions in a number of ways.
Next we single out three factors that may influence the likelihood of a
cartel of miners carrying out a particular coordinated strategy.
\begin{mdframed}[style=offset,frametitle={Classifying Coordinated Strategies}]
\begin{enumerate}

\item Is the coordinated strategy plausibly for the good of the entire
  Ethereum network, or clearly for the good of the miners?

\item Is the coordinated strategy easily detectable?

\item Is the cartel of miners game-theoretically robust (with an
  incentive for cartel members to remain) or game-theoretically fragile 
(with an incentive for members to secede, for example by
switching to solo mining or joining a competing mining pool)?

\end{enumerate}
\end{mdframed}

\subsubsection{Coordination for the Greater Good}

Ethereum miners do appear to coordinate their actions at times,
for example when resolving hard forks or increasing the
maximum block size (which in Ethereum is voted on by miners).  In
these cases, the goal is plausibly to maximize the health of the
Ethereum network.  For example, increases in the maximum block size
over time may have been a balancing act between minimizing
transaction fees and minimizing the centralization risk due to the
computation and communication necessary to process blocks.\footnote{A
  plausible alternative narrative is that miners are maximizing their
  rewards from transaction fees subject to acceptance of the maximum
  block size by the network.}

\begin{mdframed}[style=offset]
Miners appear to coordinate when the goal is
plausibly to maximize the health of the Ethereum network.
\end{mdframed}

\subsubsection{The Risk of Undetectable Coordination}

The key question is then whether miners will use this coordination
ability to pursue goals that are primarily in their own interest,
rather than in the interest of the network.  The risk is greatest from
undetectable strategies.
\begin{mdframed}[style=offset]
Miners should not be expected to eschew undetectable coordinated
strategies that are in their own self-interest.
\end{mdframed}
Protocols should therefore be designed to avoid such undetectable
strategies whenever possible, or at least to render them
game-theoretically fragile (see Section~\ref{sss:fragile}).

\begin{example}[Fake Transactions in Vickrey Auctions Are Undetectable]
Example~\ref{ex:spa_mmic}\\ notes that Vickrey (a.k.a.\ second-price)
auctions can be manipulated via fake transactions to boost a miner's
revenue.  Such manipulation could be difficult to detect, and a miner
(myopic or otherwise) would have a strong incentive to do it.  This
reinforces the argument against Vickrey auctions in permissionless
blockchains.
\end{example}

\subsubsection{The Lack of Detectable Coordinated Strategies in the Wild}\label{sss:detectable}

What about detectable coordinated strategies that favor the miners
over the network?  For example:
\begin{mdframed}[style=offset,frametitle={Three Types of Detectable
    Attacks by a 51\% Cartel}]
\begin{enumerate}

\item A double-spend attack via a significant blockchain
  reorganization.

\item A censorship attack in which every block referencing a
  blacklisted address is deliberately orphaned by the cartel.

\item A denial-of-service attack in which every non-empty block is
  deliberately orphaned by the cartel.

\end{enumerate}
\end{mdframed}
All three of these attacks have been rare to non-existent in
Ethereum.\footnote{Less secure blockchains, including Ethereum
  Classic, have suffered from such attacks~\cite{etc}.}  Why?  We can
only speculate on the reasons:
\begin{mdframed}[style=offset,frametitle={Possible Reasons Miners Avoid Detectable Attacks}]
\begin{enumerate}

\item Enough miners (or mining pools) are fundamentally opposed to
  deliberate attacks that might harm the Ethereum network, due to
  altruism or blind loyalty, that a 51\%  cartel cannot form.

\item Many miners are ETH holders and believe that a detectable attack
  would significantly decrease the price of ETH.\tablefootnote{Though three
    recent double-spend attacks on the Ethereum Classic blockchain
    have not significantly harmed the price of ETC; perhaps that
    blockchain's relatively low level of security has been priced in
    all along.  See also Moroz et al.~\cite{moroz} for further discussion.}

\item Many miners have some other form of vested interest in the
  health of the Ethereum network
and believe it would be significantly damaged by a detectable attack.
For example, an ASIC is effectively a call option on ETH~\cite{HN19,YZ20}.

\item Miners fear that they would be punished for a significant
  detectable attack through a hard fork.\tablefootnote{This fear is
    perhaps more relevant for ASIC miners than for GPU miners.}

\end{enumerate}
\end{mdframed}

\subsubsection{Game-Theoretic Fragility}\label{sss:fragile}

Perhaps miners will carry out a self-interested
coordinated strategy if and only if it is undetectable?  The monopoly
price- and quantity-setting strategies identified in
Section~\ref{ss:fpa_thought} indicate that reality is more complex.
Appropriately disguised versions of these strategies can be difficult to
detect (Remark~\ref{rem:monopoly_detect}), and yet there is little to
no anecdotal evidence suggesting that Ethereum miners have ever coordinated to
implement such strategies.  Again, there are many possible
explanations: 
\begin{mdframed}[style=offset,frametitle={Possible Reasons for the Lack of
    Price- and Quantity-Setting Strategies}]
\begin{enumerate}

\item Miners would implement such strategies if they could, but there
  are too many obstacles (e.g., rapidly changing and hard-to-predict
  demand)  to coordinating on a price or quantity for a
  significant length of time.

\item Miners would implement such strategies if it was in their
  self-interest, but typically the monopoly quantity for the current
  demand curve (Definition~\ref{def:monopoly}) equals the maximum
  block size and no deviation from the 
protocol's prescribed behavior is necessary.

\item Disguising such strategies is too
  difficult, so all the arguments against detectable strategies apply.

\item Implementing such strategies could hurt the throughput and
  therefore health of the Ethereum network, which many miners have a
  vested interest in.

\item A large cartel of miners would be game-theoretically fragile,
  with cartel members incentivized to secede.

\end{enumerate}
\end{mdframed}

To illustrate the last point, we proceed as in Houy~\cite{houy}.
Imagine that all miners belong to the
cartel and implement an optimal 100\% miner strategy, such as
price-setting at a monopoly price~$p^*$ that is larger than the
market-clearing price~$\bar{p}$.
Thus, all transactions with bid at least~$p^*$ get included in a block
while all the other transactions languish in the mempool; because $p^*
> \bar{p}$, blocks will not be full.

The optimal myopic miner strategy (Section~\ref{ss:mmic}), meanwhile,
would be to ignore the cutoff~$p^*$ and pack the current block as full
as possible with the transactions with the highest bids.
This strategy maximizes the miner's
short-term revenue rather than leaving money on the table to sustain
the upward price pressure on the creators of high-value transactions.  In
effect, such a myopic miner would be free riding on the sacrifices of
the other miners which prop up the bids of high-value transactions.  A
small miner is well approximated by a myopic miner,
so one might well expect such a miner to secede from the cartel to
implement the optimal strategy for a myopic miner rather than for a
100\% miner.

A cartel of miners could discourage secession by its members through
the threat of punishment.  For example, if the rest of the cartel
controls at least 51\% of the overall hashrate, the remaining miners
could refuse to extend any block that does not conform to its rules,
such as blocks that are packed with more than the monopoly quantity
worth of 
gas.\footnote{The ``feather forking'' variant of this strategy has
  the potential to succeed even with less than 50\% of the overall
  hashrate~\cite{feather}.}
However, there is little evidence of any Ethereum miners employing
such punishment strategies.  Why not?  Perhaps they have not been
needed.  Perhaps carrying them out would be too logistically complex.
Or perhaps punishment strategies are inevitably detectable and
therefore run into all the impediments listed in
Section~\ref{sss:detectable} for detectable coordinated strategies.

We hypothesize that coordinated strategies that harm the Ethereum
network and require a credible threat of punishment to sustain might
pose little risk. 
\begin{mdframed}[style={offset},
frametitle={Hypothesis: Game-Theoretic Fragility Is a Dealbreaker}]
Coordinated miner strategies that:
\begin{itemize}
\item [(i)] favor miners at the expense of the network;

\item [(ii)] require short-term sacrifices from each member for the
  good of the cartel; and 

\item [(iii)] are costly or difficult to sustain through punishment

\end{itemize}
may be rare in a well-secured blockchain such as Ethereum.
\end{mdframed}

\subsubsection{The Upshot}\label{sss:upshot}

The discussion in this section clarifies the most worrisome type
of miner coordination, which deserves special attention when designing
or modifying a protocol:
\begin{mdframed}[style=offset]
The most concerning type of coordinated miner strategy is one that:
\begin{enumerate}

\item is in the interest of miners, rather than the network;

\item is undetectable; and

\item is game-theoretically robust, with cartel members
  incentivized to remain.

\end{enumerate}
\end{mdframed}
Revisiting the four types of coordinated strategies discussed in this
section---double-spend attacks, censorship attacks, denial-of-service
attacks, and monopoly price- or quantity-setting---we see that the
first three attacks fail the second criterion while the fourth
strategy fails the third.

\subsection{EIP-1559: Will Miners Collude?}\label{ss:1559collusion}

We now speculate on the likelihood of different forms of miner
coordination (Section~\ref{ss:fpa_collusion}) in a post-EIP-1559
world.

First, all of the arguments in Section~\ref{sss:detectable} against
detectable attacks remain equally valid under EIP-1559, and the
specific attacks discussed (double-spend, censorship,
denial-of-service) remain equally detectable.

Second, while the optimal 100\% miner strategy is generally different
under EIP-1559 (Section~\ref{ss:1559_thought}) than under the status
quo (Section~\ref{ss:fpa_thought}) due to base fee manipulation, four
of the five potential impediments to implementing the latter
(identified in Section~\ref{sss:fragile}) apply also to the
former.\footnote{The exception is the second point.  Under EIP-1559,
  the 100\% miner strategy is generally different from the honest
  mining strategy even when the monopoly quantity exceeds the target block
  size (see Section~\ref{ss:1559_thought}).}  In particular, a cartel
simulating the optimal 100\% miner strategy under EIP-1559 is
game-theoretically fragile.  Analogous to a first-price auction, a
myopic miner is incentivized to pack its block as full as
possible with the transactions with the highest tips (up to 25M gas,
which is double the target block size).
This strategy maximizes
short-term miner revenue rather than leaving money on the table in
order to keep future base fees low.  In effect, such a myopic miner
would be free riding on the sacrifices of the other miners that
respect the target block size and thereby keep the base fee low.

Third, because of the 1559 mechanism's fee burn, fake transactions can
no longer be costlessly used to disguise an attack that simulates the
optimal 100\% miner strategy (cf., Remark~\ref{rem:monopoly_detect}
and first-price auctions).  Because this strategy is now either costly
or detectable, it is arguably even less likely to be used under
EIP-1559 than under the status quo.

\begin{mdframed}[style=offset]
The game-theoretic impediments to double-spend attacks, censorship attacks,
denial-of-service attacks, and
revenue-max\-imizing 100\% miner strategies (including base fee
manipulation) appear as strong under EIP-1559 as under the status
quo.\tablefootnote{One counterpoint is that the reduction of miner
  revenue on account of the fee burn would likely reduce the
  overall hashrate, lowering the cost to a saboteur of launching these
  attacks (e.g., by renting sufficient
  hashrate~\cite{51attacks,whybuy}).  Under EIP-1559, 
the block reward alone must be sufficiently high to incentivize an
  adequate amount of hashrate and consequent security.}
\end{mdframed}

\subsection{Caveats}\label{ss:caveats}

Sections~\ref{ss:fpa_collusion}--\ref{ss:1559collusion} show that,
with both first-price auctions and the 1559 mechanism, all of the most
concerning forms of miner collusion (double-spending, censorship,
denial-of-service, revenue-maximization) are detectable,
game-theoretically fragile, or both.  But is this really enough
evidence to conclude that the harmful effects of miner collusion will
be no worse under EIP-1559 than they are now?  This section plays
devil's advocate and suggests some possible complications.

Because of the burned base fee revenues, many miners appear to view
EIP-1559 as 
taking away some of their profits and handing them over to ETH
holders.\footnote{There is merit to this argument 
but also some counterbalancing factors.
First, because
  Ethereum mining has a relatively low barrier to entry and exit, 
a decrease in aggregate miner rewards should lead to a decrease in 
the overall hashrate, with the least profitable miners exiting (see
e.g.~\cite{EOB17,leshno}).
This, in turn, increases the relative hashrate (and corresponding
fraction of miners' rewards) of the miners who remain.  
(Ethereum security suffers as a result but remains propped up by
the block reward~\cite{auer,budish}, which EIP-1559 leaves untouched.)

Second, there is evidence that an increasing share of
Ethereum transaction fees are paid by transactions vying for special 
treatment within a block (e.g., being placed first so as to execute
prior to all other transactions in the same block)~\cite{mev}.
Provided the willingness to pay of such transactions is significantly
higher than the market-clearing price at the target block size of
12.5M gas---the quantity that the base fee proxies for---miners should
continue to collect significant fees from them through their tips in
the 1559 mechanism.}
For example, of the nine miners responding to a questionnaire by
Beiko~\cite{beiko}, six wrote that ``they would not implement it under
any circumstances.''  This strong negative reaction suggests that
EIP-1559 may galvanize miners to sustain collusion to a degree not yet
seen under the status quo.  

An immediate issue is miner adoption, and the plan for the deployment
and acceptance of EIP-1559 should be explicitly discussed.  For
example, can the Ethereum Foundation effectively dictate its use?  Or
is the plan to first secure support from major projects built on top
of Ethereum (e.g., the USDC stable coin),
thereby forcing miners' hands?  Or should further support from miners
be sought out directly, and perhaps explicitly incentivized?

A second concern is that the 1559 mechanism's fee burn could change
the norms around what types of miner collusion are culturally
acceptable (e.g., coordinating on a new maximum block size) versus
unacceptable (e.g., a censorship attack).  For example, imagine that
miners coordinated their actions to avoid the base fee but otherwise
acted as in a first-price auction with a maximum block size of 12.5M
gas.\footnote{Such coordination can be implemented with a variation of
  the strategy in Section~\ref{ss:1559collusion}: first drive the base
  fee to zero, for example by publishing a sequence of empty blocks,
  and then use the quantity-setting strategy with quantity 12.5M gas
  for all future blocks (thereby keeping the base fee at zero
  forevermore).}  Like the optimal 100\% miner strategy in
Section~\ref{ss:1559collusion}, such coordination is
game-theoretically fragile and requires each miner to leave immediate
revenue on the table that could otherwise be collected by including
more than 12.5M gas worth of transactions in a block.  On the other
hand, this coordinated strategy could be much less damaging to the Ethereum
network than
something like a censorship attack or throttling the transaction rate
to boost miner revenues.  If such a strategy was widely perceived as mostly
harmless---perhaps unlikely in this instance, given the Ethereum
community's general enthusiasm for counteracting inflation with a fee
burn---it could conceivably find more purchase among miners.

\section{Alternative Designs}\label{s:alt}

Does EIP-1559 need to work the way that it does?  Are there
alternative designs that accomplish the same goals in a better or
simpler way?

Sections~\ref{ss:refund} and~\ref{ss:fpa_burn} argue that the
seemingly orthogonal goals of easy fee estimation and fee burning are
in fact inextricably linked through the threat of off-chain agreements.
Section~\ref{ss:forward} investigates a design that pays revenue from
transaction fees forward to miners of future blocks, an alternative to
fee burning with similar game-theoretic properties.
Section~\ref{ss:beos} recaps a recent transaction fee mechanism
proposal by Basu et al.~\cite{beos}.
Section~\ref{ss:tradeoff} discusses an
alternative design that, relative to the 1559 mechanism, favors UIC
over OCA-proofness.  Section~\ref{ss:update} explores the
possibilities for alternative base fee update rules.

\subsection{Paying the Base Fee to the Miner}\label{ss:refund}

The 1559 mechanism achieves a ``good user experience,'' in the form of
a typically-UIC guarantee (Theorem~\ref{t:1559uic}).  
At first glance, the proof of
this guarantee seems to hinge on two assumptions: (i) the base fee is
determined only by past history and independent of the current block;
and (ii) the base fee is high enough that the demand for gas is at
most the maximum block size.  Where does fee burning come in?

Specifically, consider the following alternative design in which
base fee revenues are passed on to the miner of the block; we call
this the {\em 1559-R mechanism}.
(Here ``R'' stands for ``refund.'')  The allocation rule~$\allocs^R$
is identical to that of the 1559 mechanism (the rule~$\allocs^{\nine}$ in
Definition~\ref{def:1559alloc}).  Miners can no longer be counted on
to exclude transactions with bid less than the base fee, so assume
that the protocol automatically treats as invalid any transaction~$t$
in a block with a bid~$\bidt$ that is less than that block's base
fee~$r$.  The new payment rule~$\prices^R$ is identical to the payment
rule~$\prices^f$ of a first-price auction
(Example~\ref{ex:fpa_payment}), with the miner collecting the entire
bid (i.e., the minimum of the fee cap and the sum of the base fee and
tip) as revenue.  The new burning rule~$\burns^R$ is also the same as
in a first-price auction---the all-zero rule~$\burns^f$.

Unfortunately, with a simple off-chain agreement, the 1559-R mechanism
devolves into a first-price auction (with block size~25M gas, double the
target):\footnote{The base fee of the 1559-R mechanism would therefore
  eventually increase to its maximum level.}
\begin{enumerate}

\item Users bid~$r$ on-chain and communicate off-chain what
  they would have bid in a standard first-price auction.

\item If a miner includes a transaction~$t$ with off-chain
  bid~$\bidt$, $t$'s creator transfers $\bidt-r$ gwei per unit of gas
  to the miner.  (When $\bidt < r$, this should be interpreted as a
  refund of $r-\bidt$ per unit of gas from the miner to~$t$'s creator.)

\end{enumerate}
In the notation of Definition~\ref{def:oca}, this is the OCA $\oca$ in
which $\bids = \mathbf{r}$ and $\tau_t = (\bid'_t-r)$, where~$\bid'_t$
denotes what $t$'s creator would have bid in a standard first-price auction.
\begin{prop}[The 1559-R Mechanism Is Equivalent to a First-Price Auction]\label{prop:1559r}
For\\ every set of transactions and base fee, there is a
one-to-one correspondence between the outcomes possible in a
first-price auction and the outcomes possible in the 1559-R mechanism
with an OCA (with the same maximum block size).
\end{prop}

\begin{proof}
As noted above, the outcome of the bids~$\bids'$ in a first-price
auction is equivalent to the outcome of the on-chain bids~$\mathbf{r}$
under the 1559-R mechanism with off-chain transfers $\bids'-\mathbf{r}$.
In the other direction, for an outcome of the 1559-R mechanism and an
OCA in which the net payment from the creator of an included
transaction~$t$ to the miner is~$a_t$, the outcome of a first-price
auction in which the (on-chain) bids $\bids'$ are the same as
$\mathbf{a}$ (for included transactions) or~0 (for excluded
transactions) is equivalent.
\end{proof}

\begin{mdframed}[style=offset]
Transferring the revenue from the base fee of a block to the miner of
that block is economically equivalent to having no base fee.  In this
sense, a base fee provides UX improvements only if it is burned (or
otherwise withheld from the miner).
\end{mdframed}

\begin{remark}[Partial Refund of the Base Fee]
Proposition~\ref{prop:1559r} shows that, because of the possibility of
off-chain agreements, burning 0\% of the base fee is economically
equivalent to having no base fee at all.  More generally, burning an
$\alpha$ fraction of the base fee is economically equivalent to having
a fully burned base fee that is~$\alpha$ times as large.  For example,
consider a scenario in which the 1559 mechanism's base fee would
stabilize at~$r^*$.  If instead half of the base fee was burned,
one would expect the new mechanism to stabilize at a base fee
of~$2r^*$, with $r^*$ of it burned and the rest divvied up between
the miner and users via an OCA.\footnote{The designers of the NEAR
  blockchain, possibly unaware of this point, recently
  deployed a version of the 1559 mechanism in which 70\% of the
  base fee is burned and the remaining 30\% is transferred to smart
  contracts that   were used in the previous epoch~\cite{near_white}.
See Hasu~\cite{near} for further discussion.}
\end{remark}

\begin{remark}[Implementing the Monopoly Price Under the 1559-R Mechanism]\label{rem:1559r}
A second (if less important) issue with the 1559-R mechanism is that, unlike with
the 1559 mechanism, there are scenarios in which 
a coordinated miner strategy meets all three of the criteria in
Section~\ref{sss:upshot}---favoring the miners at the expense of
the network, undetectable, and game-theoretically robust.

In more detail, suppose the demand curve~$D(\cdot)$ is the same for
every block and, at the monopoly price~$p^*$, the demand (i.e.,
monopoly quantity) $D(p^*)$ is less than the target block size of
12.5M gas.  
Suppose also that the marginal
cost~$\mu$ of gas to a miner is negligible.  Miners can now simulate
the revenue-maximizing $p^*$-price-setting strategy
(Section~\ref{ss:fpa_thought}) simply by keeping the base fee at~$r^*$
at all times: (i) increase the base fee to~$r^*$, using fake
transactions as necessary; (ii) keep the base fee at~$r^*$ forevermore
by publishing 12.5M gas blocks, each including~$D(p^*)$ total gas of
real transactions and the balance in fake transactions.  This strategy
is game-theoretically robust because, given that past blocks have
resulted in a base fee of~$p^*$, a myopic miner maximizes its
immediate revenue by including all eligible transactions (i.e., those with
bid at least~$p^*$) and is not harmed by the inclusion of additional
fake transactions.
\end{remark}

\subsection{Fee-Burning First-Price Auctions}\label{ss:fpa_burn}

Alternatively, suppose we wanted a transaction fee mechanism with
a fee burn but didn't care about easy fee estimation.  Why not stick
with  first-price auctions, but burn all (or part) of the fees?  
Formally, this is the TFM $\tfm$ with $\allocs =\allocs^f,
\prices=\burns^f$, and $\burns=\prices^f$.

The problem is again the threat of off-chain agreements.
Intuitively, first-price auctions in which all payments are burned are
not OCA-proof because miners and users would be incentivized to
move all their payments off-chain.  
The next proposition formally shows that this mechanism fails to
satisfy Definition~\ref{def:ocaproof}.
\begin{prop}[Fee-Burning First-Price Auctions Are Not
  OCA-Proof]\label{prop:fpa_burn_not_ocaproof}
The fee-burn\-ing first-price auction $(\allocs^f,\burns^f,\prices^f)$
is not OCA-proof.
\end{prop}

\begin{proof}
Consider a non-empty set~$U$ of transactions with $\valt > 0$ for
every~$t \in U$, and assume that the marginal cost~$\mu$ of gas to a
miner is negligible.
Assume also that there is a unique feasible subset~$T \sse U$
of transactions maximizing the total value
\[
\sum_{t \in T} v_t \cdot g_t,
\]
and denote this maximum-possible total value by~$V > 0$.  The miner~$m$ and
users can obtain joint utility~$V$ through an OCA $\oca$ between the
creators of~$T$ and~$m$ in which $\bidt = 0$ and (for example) $\tau_t
= \valt/2$ for every $t \in T$.  

The joint utility~\eqref{eq:jointutil} of an on-chain outcome 
is~$V$ if and only if the
included transactions are precisely~$T$ and there is zero fee burn.
Every bid vector $\bids^*$ in which $\bid^*_t > 0$ for at least one
transaction~$t$ leads to a non-zero fee burn 
(on account of maximizing~\eqref{eq:fpa_obj})
and hence cannot achieve joint utility~$V$.  
Meanwhile, the all-zero
bid vector $\bids^* = \mathbf{0}$ leads to an arbitrary feasible set
$T' \sse U$ of transactions,
which is generally different than~$T$.
\end{proof}
Moreover, in the obvious OCA for the miner and users to employ in the
proof of Proposition~\ref{prop:fpa_burn_not_ocaproof},
the on-chain bids are zero and so there is no fee burn whatsoever!

\begin{mdframed}[style=offset]
Burning the fees of a first-price auction moves all payments
off-chain and leads to zero fee burning.  In this sense, a
non-trivial fee burn requires a base fee.
\end{mdframed}
\begin{remark}[Partial Fee-Burning]
The same argument and conclusion apply more generally to a
first-price auction in which any fixed positive fraction of the fees
are burned.
\end{remark}

\subsection{Paying the Base Fee Forward}\label{ss:forward}

Section~\ref{ss:refund} shows that, for a block's base fee to be
economically meaningful, revenues from it cannot be passed on to the
miner of the block.  Perhaps the simplest way to withhold this
revenue, as in the current EIP-1559 spec~\cite{1559spec}, is to
burn these revenues, effectively issuing a lump-sum refund to all ETH
holders.  An alternative solution, discussed explicitly
in~\cite{vb16}, is to transfer these revenues to one or more miners of
{\em other} blocks.

\subsubsection{The $\ell$-Smoothed Mechanism}

Concretely, consider the variant of the 1559 mechanism in which,
for some window length~$\ell$ (hard-coded into the protocol), the base
fee revenues from a block are split equally among the miners of the
next~$\ell$ blocks.  (The~1559 mechanism can be thought of as
the special case in which~$\ell=0$.)  Thus, a miner of a block receives a
$1/\ell$ fraction of the sum of the base fee revenues from the
previous~$\ell$ blocks, along with all of the tips from the current
block.

We can define the {\em $\ell$-smoothed mechanism} as follows.  
Fix a blockchain history~$B_1,B_2,\ldots,B_{k-1}$ with $k \ge \ell+1$.
Let $r_i = \alpha(B_1,B_2,\ldots,B_{i-1})$ denote the base
fee of block~$B_i$, where~$\alpha$ is the iteration
of the EIP-1559 update rule~\eqref{eq:update}.
Let~$R_k = \beta(B_1,B_2,\ldots,B_{k-1})$ denote the paid-forward base
fee revenues:
\[
\beta(B_1,B_2,\ldots,B_{k-1}) := \frac{1}{\ell} \sum_{i=k-\ell}^{k-1}
r_i \cdot G_i,
\]
where~$G_i = \sum_{t \in B_i} g_t$ denotes $B_i$'s size in gas.
The allocation, payment, and burning rules of the $\ell$-smoothed
mechanism are formally identical to those of the 1559 mechanism
(Definition~\ref{def:1559alloc}--\ref{def:1559burning}), 
with the understanding that
the burning rule (a constant function always equal to~$r_k$)
now indicates a payment that is paid forward to future miners rather
than burned.  Technically, the paid-forward base fee revenues~$R_k$
should be added to a miner's utility function
(Definition~\ref{def:mmutil}), but because~$R_k$ is independent of the
miner's current actions, it has no effect on the optimal strategy of
a myopic miner (or user).  In effect, $R_k$ serves as a fixed bonus
added to the standard block reward.

\subsubsection{Properties of the $\ell$-Smoothed Mechanism}

Because users are indifferent to how their payments are directed, and
because a myopic miner cares only about its revenue from the current
block, all of the game-theoretic guarantees for users and myopic
miners satisfied by the 1559 mechanism (Theorem~\ref{t:1559mmic},
Corollary~\ref{cor:1559costly}, Theorem~\ref{t:1559uic}, and
Theorem~\ref{t:1559ocaproof}) carry over to the $\ell$-smoothed
mechanism (for any~$\ell$).
\begin{theorem}[Guarantees for the $\ell$-Smoothed Mechanism]\label{t:smoothed}
For every $\ell \ge 0$, the $\ell$-smoothed mechanism is:
\begin{itemize}

\item [(i)] MMIC;

\item [(ii)] $(r+\mu)$-costly, where~$r$ is the current base fee
  and~$\mu$ is the marginal cost of gas;

\item [(iii)] UIC, provided the current base fee is not excessively
  low for the current demand; and

\item [(iv)] OCA-proof.

\end{itemize}
\end{theorem}
Theorem~\ref{t:smoothed} holds no matter how the base fee and
pay-forward rewards are defined (i.e., for any functions~$\alpha$
and~$\beta$).

The discussion on sustained collusion by miners under EIP-1559
(Section~\ref{s:collusion}) applies also to the $\ell$-smoothed
mechanism, with some small changes.  First, for a demand curve with
monopoly quantity more than 12.5M gas, a 100\% miner will be better
off in the $\ell$-smoothed mechanism (with $\ell \ge 1$) because it
will avoid the non-zero fee burn it would otherwise have paid (see
Section~\ref{ss:1559_thought}).  Second, the reasoning behind the
game-theoretic fragility (Section~\ref{sss:fragile}) of a cartel of
miners simulating a 100\% miner strategy is more complicated.  As
before, lost tip revenue disincentives a cartel member from
manipulating the base fee downward---the only manipulation of concern
when base fee revenues are burned.  With the base fee revenues
returned to a 100\% miner by the $\ell$-smoothed mechanism, it's now
also important that fake transactions are costly
(Theorem~\ref{t:smoothed}(ii)) to disincentive manipulations of the
base fee upward.  Finally, with base fee revenues going to miners
rather than ETH holders, the caveats in Section~\ref{ss:caveats}
become moot.

\subsubsection{Pros and Cons of the $\ell$-Smoothed Mechanism}

A basic question, worthy of lengthy debate by the Ethereum community,
is: Who should benefit from the user payments that are inevitably
generated by a fully utilized blockchain?  The fee burn in the 1559
mechanism explicitly favors ETH holders, while the $\ell$-smoothed
mechanism favors Ethereum miners.  Different stakeholders in Ethereum
will of course have their own reasons for preferring one over the
other.

A second trade-off between the 1559 and $\ell$-smoothed mechanisms
concerns whether variability in demand (and hence fees) translates to
variability in security or in the issuance of new currency.  In the
1559 mechanism, every block changes the money supply in two ways:
minting new coins for the block reward (currently 2~ETH), and burning
the coins used to pay the base fee.  Because the base fee rises and
falls with demand, Ethereum's inflation rate would be variable and
unpredictable.  On the other hand, assuming negligible tips, every
block confers roughly the same total reward to the miner (the block
reward); the security of the Ethereum network scales with this total
reward~\cite{auer,budish} and should therefore also stay
relatively constant (modulo fluctuations in the price of ETH).
Meanwhile, in the $\ell$-smoothed mechanism, inflation would be as
predictable as it is under the status quo (currently around 4\%
annually).  Instead, total miner reward would vary with the revenue
generated by the base fee, leading to an unpredictable level of
security (though never less than that with the 1559 mechanism).

Finally, because of its variable total reward, the $\ell$-smoothed
mechanism is vulnerable to certain attack vectors that would be
fruitless under the 1559 mechanism, especially when~$\ell$ is small.
For example, imagine that~$\ell=1$ and a miner~$m_1$ mines a
block~$B_1$ with an unusually large sum~$R$ of transaction fees.
This windfall would be reaped by the miner~$m_2$ 
of the next block~$B_2$; suppose further that the sum of transaction
fees in~$B_2$ is much less than~$R$.  At this juncture, a miner~$m_3$
might consider trying to extend~$B_1$ with a block~$B_3$ in order to
orphan~$B_2$; if other miners happen to extend~$B_3$ rather
than~$B_2$, $m_3$ will effectively have stolen the reward of~$R$
from~$m_2$.\footnote{This is similar to the undercutting attack
  of~\cite{ccs16} for a regime in which transaction fees dominate
  block rewards; see Section~\ref{ss:benefits} for further
  discussion.}
Such examples suggest choosing a large value of~$\ell$ (e.g.,
$\ell=1000$) to guarantee that consecutive blocks will have nearly
identical total rewards associated with them.

\begin{remark}[A Blended Mechanism]
The 1559 and $\ell$-smoothed mechanisms can be easily
blended to balance the competing concerns of miners and ETH holders
and the variability in issuance and security.  For example, for a
parameter $\lambda \in [0,1]$, a mechanism could burn a $\lambda$
fraction of the base fee revenues and pay forward the remaining
$1-\lambda$ fraction.  Theorem~\ref{t:smoothed} and the subsequent
discussion on miner collusion remain valid for such blended
mechanisms.
\end{remark}

\subsection{The BEOS Mechanism}\label{ss:beos}

A variant of the ``pay-it-forward'' design philosophy in
Section~\ref{ss:forward} was proposed also by Basu et al.~\cite{beos}
for a transaction fee mechanism that is not directly related to the
1559 mechanism.  We next explain a slightly simplified version of
their proposal, which we call the {\em BEOS mechanism} (after
its proposers).

There is a fixed block size, say 12.5M gas, and no base fee.  The
first key idea is to charge all transactions included in a block a
common price (per unit of gas), namely the lowest bid of an included
transaction.  Miner revenue is then the block size (in gas) times the
lowest bid of an included transaction, and so a revenue-maximizing miner 
may exclude transactions in order to boost the lowest included
bid.\footnote{This is exactly the ``monopolistic price'' mechanism
  proposed by Lavi et al.~\cite{LSZ19}; they were motivated by the
  problem of maximizing the security provided by transaction fees (at
  the expense of economic efficiency) in a future in which Bitcoin's
  block rewards are negligible.  This mechanism is MMIC; is
  ``approximately UIC,'' in the sense that truthful bidding is an
  approximately dominant strategy for users as the number of
  users grows large~\cite{LSZ19,Y18}; and is not
  OCA-proof (on account of failing to maximize the joint utility of
  the miner and users).}  For example, for a block with room for three
transactions and a mempool containing three transactions with bids 10,
8, and 3, a revenue-maximizing miner would include the
first two transactions while excluding the third (to earn revenue $2
\times 8 = 16$).  (Cf., Example~\ref{ex:spa_mmic}.)

The second key idea is to automatically charge only a minimum
transaction fee---for example, just enough to cover the marginal
cost~$\mu$ of gas---to all transactions in any block that is not
(almost) full.  This rule by itself is toothless and leads to an
equivalent mechanism, as a miner can costlessly extend its favorite
underfull block with minimum bid~$b$ to a full block with minimum
bid~$b$ using fake transactions (all with bid~$b$).

The final key idea in the BEOS mechanism is to pay transaction fees
forward, with the transaction fee revenue from a block~$B$ split
evenly between~$B$'s miner and the miners of the~$\ell-1$ subsequent blocks.
Thus, the miner of a block gets a $1/\ell$ fraction of the transaction
fee revenue in that block, along with a $1/\ell$ fraction of the
combined revenue of the preceding $\ell-1$ blocks.  As a result, for
$\ell \ge 2$, fake transactions now carry a cost: the miner pays their
full transaction fees but recoups only a $1/\ell$ fraction of them as
revenue.

The BEOS mechanism is arguably simpler than that proposed in EIP-1559,
as there is no base fee to keep track of.  Its game-theoretic
guarantees are considerably weaker, however.  While the ``pay it
forward'' idea helps discourage fake transactions, the BEOS mechanism
is not in general MMIC.\footnote{Basu et al.~\cite{beos} prove that
  the mechanism becomes ``approximately MMIC'' in the case of a very
  large number of transactions with i.i.d.\ valuations drawn from a
  distribution with bounded support.}  It is ``approximately UIC'' as
the number of users grows large, in the sense that no bidding strategy
generates significantly more utility than truthful bidding.  It is not
OCA-proof (for~$\ell \ge 2$), for the same reasons that a first-price
auction with fee burning is not OCA-proof
(Proposition~\ref{prop:fpa_burn_not_ocaproof}).
Thus, from a game-theoretic perspective, the BEOS mechanism does not
appear competitive with the 1559 mechanism.

\subsection{The Tipless Mechanism: Trading Off UIC and OCA-Proofness}\label{ss:tradeoff}

The 1559 mechanism uses tips to achieve OCA-proofness in all blocks
(Theorem~\ref{t:1559ocaproof}), at the expense of losing the UIC
condition in blocks with excessively low base fees (see
Section~\ref{sss:1559uic}).  This section presents an alternative
design with the opposite trade-off---one that is always UIC, and
OCA-proof except in blocks with an excessively low base fee.

\subsubsection{The Tipless Mechanism}

We next define the {\em tipless mechanism}, so-called because it is
essentially the 1559 mechanism with constant and hard-coded tips
rather than variable and user-specified tips.  As with the 1559
mechanism, each block has a base fee $r =
\alpha(B_1,B_2,\ldots,B_{k-1})$ that depends on past blocks and is
burned (or alternatively, paid forward as in Section~\ref{ss:forward}).
The creator of a transaction~$t$ specifies a fee cap~$c_t$
but no tip.  This parameter induces a bid $\bidt$ for the transaction
with respect to any given base fee~$r$, namely
\begin{equation}\label{eq:tipless_bid}
\bidt = \min\{ r + \delta, c_t \}.
\end{equation}
Here $\delta$ is a hard-coded tip to incentivize miners to include
transactions---for example, equal to (or perhaps slightly higher than)
the marginal cost~$\mu$ of gas to miners.\footnote{More
  generally, the hard-coded tip~$\delta$ could be adjusted over time in
  the same way as the block reward, through social
  consensus and hard forks.}
The only difference between the tipless mechanism and the 1559
mechanism is the number of user-specified parameters and their
interpretation as bids relative to the current base
fee---that is, the types of bidding strategies available to users.  
The allocation, payment, and burning rules of the tipless
mechanism are formally identical to those of the 1559 mechanism (Definitions~\ref{def:1559alloc}--\ref{def:1559burning}).
Given that all the tips are the same and cover a miner's marginal cost
of gas, the allocation rule~\eqref{eq:1559_obj} boils down to
packing a block as full as possible with transactions~$t$ with a
bid~$\bidt \ge r+\delta$.  The creator of an included transaction
pays~$r+\delta$ gwei per unit of gas, of which~$r$ is burned
and~$\delta$ is transferred to the miner.\footnote{This variant of the
  1559 mechanism has been implemented in the NEAR
  protocol~\cite{near_white}; see Hasu~\cite{near} for further
  discussion.}

\subsubsection{Properties of the Tipless Mechanism}

The proof that the 1559 mechanism is incentive compatible for myopic
miners (MMIC)
depends only on the form of the allocation, payment, and burning rules
of the mechanism; it is agnostic to the process by which transactions'
bids are set (see Theorem~\ref{t:1559mmic}).
Thus, the same proof applies equally well to the
tipless mechanism.
\begin{theorem}[The Tipless Mechanism is MMIC]\label{t:tipless_mmic}
The tipless mechanism is MMIC.
\end{theorem}

Now consider a block in which the base fee~$r$ is excessively low
(Definition~\ref{def:low}), meaning that the demand for gas at
price~$r+\delta$ is more than the maximum block size~$G$.  In the 1559
mechanism, the creators of transactions willing to pay at
least~$r+\delta$ must then compete for inclusion through their tips.
As a result, analogous to a first-price auction
(Example~\ref{ex:fpa_uic}), in this case the mechanism is not
incentive compatible for users (UIC)---there are no ``obvious optimal
parameters'' to associate with a transaction.

In the tipless mechanism, such transaction creators do not have the
vocabulary to differentiate themselves by offering to pay extra.  As a
result, the mechanism remains UIC even in blocks with an excessively
low base fee.
\begin{theorem}[The Tipless Mechanism is UIC]\label{t:tipless_uic}
The tipless mechanism is UIC.
\end{theorem}

\begin{proof}
Fix an on-chain history $B_1,B_2,\ldots,B_{k-1}$ and corresponding
base fee $r = \alpha(B_1,B_2,\ldots,B_{k-1})$, and a set~$T$ of
transactions.  
Suppose the creator of a transaction $t \in T$ sets its fee cap equal
to its maximum willingness to pay, corresponding to the bid
\begin{equation}\label{eq:dse}
\epnebid(\valt) = \min \{ r+\delta, \valt \}.
\end{equation}
Could some other bid be better?  
For a low-value transaction (with $\valt < r + \delta$), every
alternative  bid $\hat{\bid}_t$ either has no effect on $t$'s utility or
leads to~$t$'s inclusion in the block; the latter only occurs when
$\hat{\bid}_t \ge r+\delta$, in which case the creator's utility drops
from~0 to $(\valt -\hat{\bid}_t) \cdot g_t < 0$.  
For a high-value transaction (with $\valt \ge r + \delta$), every
alternative  bid $\hat{\bid}_t$ either has no effect on the creator's
utility or, if the alternative bid triggers~$t$'s exclusion, drops its
utility from a nonnegative number $(\valt - r- \delta) \cdot g_t \ge
0$ to~0.\footnote{If~$r$ is an excessively low base fee and the demand at
  price~$r+\delta$ is more than the maximum block size, the miner
  maximizes its revenue by packing its block as full as
  possible (as the tip-per-unit-gas~$\delta$ is the same for every
  transaction).  That is, the miner
includes the feasible set of transactions that
 maximizes the total gas used.  We assume that, if there is a tie
 between two or more   such feasible sets, the miner breaks the tie in a
  consistent way, independent of transactions' fee caps.}
We conclude that the bid in~\eqref{eq:dse} is always
utility-maximizing for~$t$'s creator.\footnote{Moreover, the bidding strategy
  $\epnebid(\cdot)$ is a symmetric dominant-strategy
    equilibrium in the sense of footnote~\ref{foot:dse}.
That is, the suggested bid~$\epnebid(\valt)$ is utility-maximizing for
$t$'s creator no matter what the other bids are
(i.e., even if the other bids differ from those suggested by the
strategy~$\epnebid(\cdot)$).\label{foot:dse2}}
\end{proof}

Further, the tipless mechanism is OCA-proof except during periods of
rapidly increasing demand.
\begin{theorem}[The Tipless Mechanism Is Typically OCA-Proof]\label{t:tipless_ocaproof}
Fix an on-chain history $B_1,B_2,\ldots,B_{k-1}$ and corresponding
base fee $r = \alpha(B_1,B_2,\ldots,B_{k-1})$, and a set~$T$ of
transactions for which~$r$ is not excessively low.  
With~$\delta = \mu$, the tipless mechanism is OCA-proof.
\end{theorem}

\begin{proof}
The joint utility~\eqref{eq:jointutil} of the miner and users for the
current block~$B_k$ is
\begin{equation}\label{eq:tipless}
\sum_{t \in B_k} ( \valt - r - \mu) \cdot g_t.
\end{equation}
Because~$r$ is not excessively low for~$T$, the total gas consumed by
transactions~$t$ with $\valt \ge r + \mu$ is at most the maximum block
size~$G$.  The joint utility~\eqref{eq:tipless} is therefore maximized
by including precisely these transactions.
This outcome can be achieved on-chain (with~$\bidt = \min \{r
+\mu,\valt \}$ for each $t \in T$), and thus cannot be improved upon
by an OCA.
\end{proof}

\begin{remark}[The Tipless Mechanism Is Not Always OCA-Proof]
The tipless mechan\-ism is not generally OCA-proof when the base fee~$r$ is
excessively low (even with $\delta = \mu$).
In this case, a miner is instructed by the allocation rule to pack its
block as full as possible using transactions with bid at
least~$r+\mu$.  With an excessively low base fee,
the feasible subset of such transactions that maximizes the block size
$\sum_{t \in T} g_t$ is generally different from the feasible subset that
maximizes the joint utility $\sum_{t \in T} (v_t-r-\mu) \cdot g_t$.\footnote{For
  example, with~$G=2$ and $r+\mu=1$, consider one eligible
  transaction with $\valt = 2$ and $g_t = 1$ and another with $\valt =
  1$ and $g_t = 2$.}  The miner and users can then strictly increase
their joint utility with an OCA that instead includes the latter
subset of transactions (for example, with transfers arranged
to share the increase in joint utility equally among the miner and
users).
\end{remark}

\subsubsection{Pros and Cons of the Tipless Mechanism}

Perhaps the strongest argument in favor of the tipless mechanism over
the 1559 mechanism is its simplicity.  On the user side, there are
several simplifications.  The creator of a transaction~$t$ only has to
specify one parameter (a fee cap~$c_t$) rather than two (a fee
cap~$c_t$ and a tip~$\delta_t$).  The ``obvious optimal bid'' in the
tipless mechanism (setting~$c_t=\valt$) is optimal for every block and
no matter what the bids of the competing transactions.
The ``obvious optimal bid'' in the 1559 mechanism
(setting~$c_t=\valt$ and~$\delta_t=\mu$) is optimal only in blocks
without an excessively low base fee, and only after assuming that
other transactions' bids were set in the same way.  On the miner side,
the revenue-maximizing strategy simplifies to maximizing the block
size while using only transactions with a bid that is at least~$r+\mu$
(where~$r$ denotes the current base fee).  Relatedly, miners have no
levers by which to pressure users to increase their tips (cf., footnote~\ref{foot:mediumsize}).

What about the mechanism's drawbacks?
First, the hard-coded tip~$\delta$ is yet another somewhat arbitrary 
parameter than may need to be adjusted over time through network
upgrades.\footnote{Possible counterargument: with so many such
  parameters already (e.g., opcode gas costs~\cite{yellowpaper}),
  what's one more?}
Second, when there are blocks with excessively low base fees (due to
rapidly increasing demand), OCA-proofness breaks down.  At such times,
one might expect miners and users to simulate the on-chain tips of the
1559 mechanism with an off-chain agreement.  Even with a base fee that
is not excessively low, such agreements might be used to accommodate
transaction creators angling for a specific block position (as opposed
to mere inclusion).\footnote{Depending on how miners choose to break ties among
eligible transactions for inclusion in such a block, on-chain
shenanigans may also be possible (e.g.~\cite{random_priority}).}

\subsection{Alternative Base Fee Update Rules}\label{ss:update}

The game-theoretic guarantees for the 1559 mechanism
(Sections~\ref{s:mm}--\ref{s:collusion}) and inseparability of easy
fee estimation and fee withholding
(Section~\ref{ss:refund}--\ref{ss:forward}) argue strongly for a
history-dependent base fee, the revenues from which are burned or
otherwise withheld from a block's miner.  Accordingly, in this section
we consider only designs with such a base fee.

But how, exactly, should the base fee be computed from the
blockchain's history?  The MMIC
(Theorem~\ref{t:1559mmic}), typically-UIC (Theorem~\ref{t:1559uic}),
and OCA-proof (Theorem~\ref{t:1559ocaproof}) guarantees from
Section~\ref{s:mm} hold no matter how the base fee is set.  The
impediments to miner collusion identified in Section~\ref{s:collusion}
likewise give little guidance as to how the base fee should evolve
over time.  The goal of this section is to clarify the assumptions
baked into the update rule in the current EIP-1559 spec~\eqref{eq:update}
and identify a few axes along which to experiment.

\subsubsection{Assessing Update Rules}\label{sss:assess}

The ideal base fee for a block is the market-clearing price for the
current mempool and block size (see Section~\ref{ss:market}).  An
ideal base fee update rule would magically guess this price,
immediately adjusting to sudden changes in demand.  A good base fee
update rule should reasonably approximate this magical one, without
introducing any undue incentives for base fee manipulation by miners
and users or vulnerabilities to outside attacks.
\begin{mdframed}[style=offset,frametitle={Desiderata for a Base Fee Update Rule}]
\begin{enumerate}  

\item Adjusts upward reasonably quickly after a sudden spike in demand.

\item Adjusts downward reasonably quickly after a sudden drop in
  demand.

\item Adjusts slowly enough to avoid overreacting to small or very
  short-lived changes in demand.

\item Cannot be manipulated by a cartel of users and/or miners in a
  game-theoretically robust way (cf., Section~\ref{sss:fragile}).

\item Expensive for an attacker to exploit.

\end{enumerate}
\end{mdframed}
How quickly is ``reasonably quickly''?  How expensive is
``expensive''?  Such questions are outside the scope of this report
and best answered through experimentation and community discussion.
The rest of this section assesses the update rule in the EIP-1559 spec
according to these criteria and suggests some alternatives to explore.

\subsubsection{Decomposable Update Rules}

In principle, the base fee~$r$ for a block~$B_k$ can be an
arbitrary function~$\alpha$ of the blockchain history
$B_1,B_2,\ldots,B_{k-1}$:
\[
r = \alpha(B_1,B_2,\ldots,B_{k-1}).
\]
In practice, however, the base fee should not be overly burdensome to
compute.  This point motivates the next definition.
\begin{definition}[Decomposable Update Rule]\label{def:decompose}
An update rule~$\alpha$ is {\em decomposable} if it can be written
\[
\alpha(B_1,B_2,\ldots,B_{k-1}) = \zeta(B_{k-1})  \cdot \alpha(B_1,B_2,\ldots,B_{k-2}),
\]
where~$\zeta$ is the {\em adjustment function}.
\end{definition}
Definition~\ref{def:decompose} encodes two different restrictions.
First, the base fee of a block should depend on only the base fee and
the contents of the most recent block.  Second, the adjustment
function~$\zeta$ depends only on the contents of the most recent
block~$B_{k-1}$ and not on its base fee.
\begin{example}[The EIP-1559 Update Rule Is Decomposable]\label{ex:1559zeta}
The update rule~\eqref{eq:update} in the EIP-1559 spec is decomposable
with 
\begin{equation}\label{eq:1559zeta}
\zeta(B_{k-1}) = 1+
\frac{1}{8} \cdot \left( \frac{g(B_{k-1}) - G_{target}}{G_{target}}
\right),
\end{equation}
where $g(B) = \sum_{t \in B} g_t$ denotes the size (in gas)
of block~$B$ and~$G_{target}$ a target block size (e.g., 12.5M gas).
\end{example}
A base fee computed by a decomposable update rule can be expressed in
a compact product form.
\begin{prop}[Product Form for Decomposable Update Rules]\label{prop:product}
If~$\alpha$ is a decomposable update rule with adjustment
function~$\zeta$ and~$r_0$ is the base fee of the genesis block~$B_1$,
then for every blockchain history~$B_1,B_2,\ldots,B_{k-1}$,
\begin{equation}\label{eq:product}
\alpha(B_1,B_2,\ldots,B_{k-1}) = r_0 \cdot \prod_{i=1}^{k-1} \zeta(B_i).
\end{equation}
\end{prop}

Non-decomposable update rules are more complex but could potentially
be useful.  For a reasonably natural example, suppose we wanted to
limit the lifetime over which any given block affects the base fee:
\begin{example}[Sliding Windows Are Not Decomposable]
Consider a base fee update rule that depends on only the most recent
$\ell$ blocks, for some parameter~$\ell$ (e.g., 100 or 1000):
\[
\alpha(B_1,B_2,\ldots,B_{k-1}) = r_0 \cdot \prod_{i=k-\ell}^{k-1} \zeta(B_i),
\]
where $k$ is assumed to be at least~$\ell+1$, and~$r_0$ and~$\zeta$
denote the initial base fee and adjustment function, respectively.
Because the
change in base fee depends on both the block entering the sliding
window (the most recent one) and the block exiting this window
(from~$\ell+1$ blocks back), this update rule is not decomposable.
\end{example}

\begin{remark}[Oscillatory Behavior of Decomposable Update Rules]
M.\ Ferreira, D.\ Moroz, and M.\ Stern (personal communication,
October 2020) point out that, in certain pathological scenarios,
decomposable update rules can oscillate between two base fees rather
than converge to a market-clearing base fee, even during a period of
stable demand.  For example, consider
such a rule with an
adjustment function~$\zeta$ satisfying $\zeta(B) = \tfrac{3}{2}$
for maximum-size blocks~$B$ and $\zeta(B) = \tfrac{2}{3}$ for empty
blocks~$B$.  Suppose the current base fee is~$r$ and there is a huge
mempool of transactions, the fee caps of which all happen to land in the
interval $[1.1r, 1.4r]$.  (Assume that all tips are negligible.)  What
happens next?  

Because all transactions are willing to pay the current
base fee of~$r$, $r$ is an excessively low base fee and the next miner
will produce a maximum-size block.  As a
result, the base fee will jump from~$r$ to~$\tfrac{3}{2}r$, at which
point no transactions are willing to pay the base fee!  The next miner
has no choice but to produce an empty block, and the base fee will
return to~$r$.  This oscillation between the base fees~$r$
and~$\tfrac{3}{2}r$, and between maximum-size and empty blocks, could
in principle continue forever.\footnote{With the adjustment function in the
  current EIP-1559 spec~\eqref{eq:1559zeta}, such an oscillation will stop
  eventually, although possibly only after a large number of blocks
  (depending on how tightly concentrated transactions' bids are).}

Such oscillatory behavior may be unlikely in a real deployment, given
the variety of tools and considerations likely to be used when
specifying the bidding parameters for a transaction.  If necessary,
Ethereum clients could explicitly inject randomness into these
parameters to avoid such pathological outcomes.
\end{remark}

Having noted that non-decomposable update rules may be worth
experimenting with, we now narrow our focus to decomposable rules.

\subsubsection{What Should the Adjustment Function Depend On?}

Designing a decomposable update rule boils down to designing its
adjustment function~$\zeta$.  By assumption, this function depends
only on the contents of the most recent block~$B_{k-1}$.  In
principle, the function $\zeta(B)$ could depend on $B$'s contents in
arbitrarily complex ways.  In the EIP-1559 adjustment
function~\eqref{eq:1559zeta}, $\zeta(B)$ depends only on the total gas
$g(B) = \sum_{t \in B} g_t$ used in~$B$, and not on any finer-grained
information about its transactions.

While it's easy to imagine alternative adjustment functions, care must
be taken with the incentives.  As a cautionary tale, consider an
adjustment function~$\zeta$ that tries to do away with
variable-size blocks through its dependence on the bids
attached to the transactions in a block~$B$.  
\begin{example}[Incorporating Bids into the Adjustment Function]\label{ex:bad}
In this design, the
target block size and the maximum block size are the same (e.g., 12.5M
gas).  If a block~$B$ has size less than the maximum, then the
adjustment function satisfies $\zeta(B) < 1$
and the base fee decreases for the next block (as in EIP-1559).  For a
full block~$B$, the adjustment function considers the minimum (or
average, or median, or\ldots) tip of a transaction in the block.  If
this statistic is close to~0, the base fee remains unchanged
($\zeta(B) = 1$); if it's significantly larger than~0, the base fee
increases ($\zeta(B) > 1$).  

The problem?  When the current base fee is excessively low, there is
no disincentive to the users or miners from colluding to keep it low.
For suppose miners and users moved the tip market off-chain, similar
to the proof of Proposition~\ref{prop:fpa_burn_not_ocaproof}.  Users
and miners are indifferent to whether payments are on- or off-chain,
as the fee burn and gas costs are the same either way.  But now the
on-chain tips are all~0 and the base fee will not increase.
\end{example}

\begin{remark}[OCA-Proofness vs.\ Miner Collusion]
The off-chain agreement in Example~\ref{ex:bad} does not violate
OCA-proofness (users and miners are equally well off with the OCA, 
but not strictly better off) and hence does not contradict  the
aforementioned fact that the 1559 mechanism remains OCA-proof no 
matter how its base fee is computed.
However, the OCA in Example~\ref{ex:bad} does show that this variant
of the 1559 mechanism encourages the most concerning type of 
coordinated miner strategy (Section~\ref{sss:upshot})---one that
favors the miners at the expense of the network, is potentially
undetectable, and is game-theoretically robust.
\end{remark}
More generally, Example~\ref{ex:bad} illustrates how OCAs can be used
to manipulate any attempt to incorporate the bids attached to a
block's transactions into the adjustment function.  Given these
dangers, it is unsurprising that the adjustment function in EIP-1559
depends only on the gas consumed by the included transactions, and
it is unclear if any additional information could be safely used.

\subsubsection{The Functional Form of the Adjustment Function}

Even after committing to an adjustment function that is a function
solely of the block size there remains flexibility in the function's
form and parameters.  The choices of these in the EIP-1559 adjustment
function~\eqref{eq:1559zeta} appear fairly arbitrary and are prime
candidates for experimentation; we next offer some possible
alternatives.

By assumption, we are now considering update rules of the form
$\zeta(B) = f(g(B))$, where~$f$ is a univariate real-valued function
and $g(B) = \sum_{t \in T} g_t$ denotes the size of block~$B$.
Only nondecreasing functions are sensible choices for~$f$---big
blocks suggest excess demand and that the base fee should be
increased, small blocks that it should be decreased.  Any continuous
such function~$f$ effectively has a ``target block size,'' 
meaning a gas threshold~$G_{target}$ such that $f(G_{target}) = 1$.

What functional form should~$f$ have?  The current EIP-1559
spec uses an adjustment function~\eqref{eq:1559zeta} with the form 
\begin{equation}\label{eq:f}
f(x) = 1 + h(x),
\end{equation}
where~$h$ is an increasing linear function with $h(G_{target})=0$.
Linear functions are attractive for their simplicity, but a nonlinear
function~$h$ might well strike a better balance between the competing
goals listed in Section~\ref{sss:assess}.  

V.\ Buterin (personal communication, October 2020) suggests an
alternative functional form, motivated by the fact that
the function $1+x$ is well approximated by~$e^x$ 
when~$x$ is small (where $e=2.718\ldots$ is Euler's number):
\begin{equation}\label{eq:exp}
f(x) = e^{h(x)},
\end{equation}
where~$h$ is an increasing function equal to~0 at $G_{target}$.
Decomposable update rules with an adjustment function of this form are
especially aesthetically appealing when written in product
form~\eqref{eq:product}:
\[
\alpha(B_1,B_2,\ldots,B_{k-1}) 
= r_0 \cdot \prod_{i=1}^{k-1} \zeta(B_i)
= r_0 \cdot \prod_{i=1}^{k-1} e^{h(g(B_i))}
= r_0 \cdot \exp \left\{ \sum_{i=1}^{k-1} h(g(B_i)) \right\}.
\]
For example, plugging in the function $h(x) =
(x-G_{target})/8G_{target}$ used in the current EIP-1559
spec:
\begin{eqnarray*}
\alpha(B_1,B_2,\ldots,B_{k-1}) 
& = & r_0 \cdot \exp \left\{ \frac{1}{8} \sum_{i=1}^{k-1}
  \left( \frac{g(B_i)-G_{target}}{G_{target}} \right) \right\}\\
& = & r_0 \cdot \exp \left\{ \frac{1}{8} \left( \frac{\sum_{t \in B_1 \cup
        \cdots \cup B_{k-1}} g_t}{G_{target}} -(k-1) \right) \right\}.
\end{eqnarray*}
The final expression makes clear that this base fee depends only on
the amount of gas consumed to date (along with the block height~$k$
and initial base fee~$r_0$), and not on how this gas was distributed
across the past blocks.  The adjustment function~\eqref{eq:1559zeta}
proposed in EIP-1559 does not have this property; for example, two
blocks with size equal to the target leave the base fee unchanged,
while an empty block followed by a block with size double the target
(or vice versa) have the cumulative effect of multiplying the base fee
by~$\tfrac{63}{64}$.

\begin{remark}[Compromising with Taylor Approximations]
Exponential functions are less convenient numerically than
polynomials.  The adjustment function in~\eqref{eq:f} can be viewed as
a degree-1 polynomial approximation of the exponential adjustment
function~\eqref{eq:exp}.  A natural compromise is to instead use the 
degree-2 polynomial approximation suggested by the exponential
function's Taylor series:
\[
f(x) = 1 + h(x) + \frac{h(x)^2}{2}.
\]
For example, plugging in the function $h(x) =
(x-G_{target})/8G_{target}$ gives a novel adjustment function:
\[
\zeta(B) = 1 + \frac{1}{8} \cdot \frac{g(B)-G_{target}}{G_{target}} +
\frac{1}{128} \frac{(g(B)-G_{target})^2}{G_{target}^2}.
\]
\end{remark}

\subsubsection{Choosing the Rate of Change}\label{sss:chooserate}

One ``magic number'' that jumps out from the adjustment
function~\eqref{eq:1559zeta} proposed in EIP-1559 is the factor
of~$\tfrac{1}{8}$, which controls how rapidly the base fee can change
from one block to the next.  More generally, an important design
question is the minimum and maximum values that an adjustment
function~$\zeta(B)$ can take on (in~\eqref{eq:1559zeta}, $\tfrac{7}{8}$
and $\tfrac{9}{8}$, respectively).  The goal should be to strike a
balance between the desiderata listed in Section~\ref{sss:assess}.

The factor of~$\tfrac{1}{8}$ in~\eqref{eq:1559zeta} means that a
sequence of maximum-size blocks (with double the target size) would
double the base fee in under 1.5 minutes (assuming one new block on
average every 13--15 seconds~\cite{blocktime}) and increase it by an
order of magnitude in under~5 minutes.  A sequence of empty blocks
would decrease the base fee at a similar, slightly faster, rate.
Thus, for demand shocks that persist for tens of minutes or more, the
base fee should have sufficient time to adjust.  The base fee would
not respond much to short-lived demand shocks, although sudden demand
increases would be mitigated by the additional throughput offered by
variable-size blocks (cf., Example~\ref{ex:traj}).  Overall, for balancing
the first three desiderata in Section~\ref{sss:assess}, the initial
choices of the factors~$\tfrac{7}{8}$ and~$\tfrac{9}{8}$ for the
minimum and maximum change in base fee seem as good as any.  However,
this design choice should clearly be revisited after there is more
data from experiments with and deployments of the 1559
mechanism.\footnote{For example, the factor of~$\tfrac{1}{8}$ could be
  added to the list of hard-coded parameters whose values are
  revisited with every network upgrade, joining the block reward,
  opcode gas costs, and so on.}

A different principled way to derive a maximum rate of change for the
base fee is to consider an attacking cartel of miners that strives to
overwhelm the network with a sequence of maximum-size blocks 
(cf., the fifth goal in
Section~\ref{sss:assess}).  For example, consider the adjustment function
in~\eqref{eq:1559zeta} and suppose that the minimum base fee is 1 gwei
and the maximum block size is 25M gas.  Five minutes of such a
``double-full block attack,'' starting from the minimum-possible base
fee and assuming that all blocks during this period are mined by the
cartel, would typically cost at least
\[
\underbrace{25000000}_{\text{max block size (gas)}} \times \quad
\sum_{i=1}^{20} \underbrace{\left( \frac{9}{8}
  \right)^{i-1}}_{\substack{\text{base fee of}\\ \text{$i$th block (gwei)}}}
\approx \text{1.9 ETH};
\]
thirty minutes would cost roughly
\[
25000000 \times \quad
\sum_{i=1}^{120} \left( \frac{9}{8}
  \right)^{i-1} \approx \text{275000 ETH},
\]
or roughly 165 million USD at an exchange rate of 600 USD/ETH; and so
on.  Similar calculations can be used to reverse engineer an
appropriate maximum rate of base fee change from a
target cost for a double-full block attack of a given
duration.

\begin{remark}[Variable Block Sizes vs.\ Variable Rate of Block
  Creation]
Short (e.g., five-minute) double-full block attacks
  appear unlikely to significantly harm the Ethereum network, provided
the existing vulnerabilities to adversarially constructed
blocks~\cite{broken_metre} are addressed.
A sequence of~$n$ double-full blocks in a given time period imposes
roughly the same load on the network as~$2n$ target-size blocks during
the same period.  Because blocks are effectively created  by a Poisson
process rather than deterministically, the Ethereum network must
already accommodate short periods during which the gas consumption is
double its expectation.\footnote{And with the proof-of-stake design in ETH 2.0, the rate of block
creation will be roughly deterministic; there, the new variability in
block sizes under EIP-1559 will effectively be canceled out
by the variability eliminated from the rate of block creation.}
\end{remark}

\subsubsection{Choosing the Block Elasticity}\label{sss:elastic}

A second ``magic number'' in EIP-1559 is the ratio of~2 between the
maximum and target block sizes.  Holding the target block size fixed,
why not a larger maximum block size?  Or a smaller one?

For flexibility and to absorb short and sudden demand spikes (cf.,
Example~\ref{ex:traj}), a bigger maximum block size is better.  The
problem with a big maximum block size is the computation and
bandwidth required by full nodes to process blocks, and the
consequent risks of greater centralization.
A ratio of~2 is one simple compromise
between these two competing forces.

A ratio of~2 between the maximum and target block size is also
convenient because only a 51\% cartel of miners could significantly
manipulate the base fee or the long-run average block size.  (For
example, with a 49\% cartel, the non-colluding
miners can negate maximum-size blocks with empty blocks and vice
versa.)
With a ratio of only~$\tfrac{3}{2}$, say, one of two compromises must
be struck: (i) leave the base fee adjustment function as
in~\eqref{eq:1559zeta}, in which case a 34\% cartel could
manipulate the base fee downward (as it would now take two maximum-size
blocks to negate an empty block produced by the cartel); or (ii) make
the adjustment function in~\eqref{eq:1559zeta} asymmetric so that empty and
maximum-size blocks continue to negate each other, in which case a
34\% cartel could reduce the long-run average block size to less than the
target block size, thereby reducing throughput (again by producing empty
blocks).  Ratios bigger than~2 seem less problematic, as a 34\% cartel
of miners would presumably not want to manipulate a burned base fee
upward.

Overall, these points suggest taking the ratio between the maximum
and target block size as large as possible, subject to the network
having the computational resources to process a short burst of
maximum-size blocks.  The ``best'' choice of this parameter may
evolve over time, and could be added to the list of parameter choices 
that are revisited with each network upgrade.

\section{Additional Remarks}\label{s:add}

\subsection{Side Benefits of EIP-1559}\label{ss:benefits}

This report assesses the transaction fee mechanism proposed in
EIP-1559 from the perspective of easy fee estimation for Ethereum
users (formalized by the ``typically-UIC'' guarantee of
Theorem~\ref{t:1559uic}).  
Several byproducts of the design are of value in their own right. 

First, as we observed in Section~\ref{ss:toohigh}, 
easy fee estimation and the introduction of
variable block sizes should decrease the
variance in transaction fees during
periods of changing demand.

Second, EIP-1559 introduces fee burning through its burned base fee.
Fee burning (or otherwise withholding base fee revenues from a block's
miner) is necessary for the base fee to be economically meaningful
(see Section~\ref{ss:refund}), but arguably is a ``necessary good''
rather than a ``necessary evil.''  Ethereum's current rate of
inflation---due to block, uncle, and nephew rewards---is roughly
4\%.  If transaction fees continue to be high, and a significant
portion of them are burned, the inflation rate will decrease and
could even turn negative.\footnote{For example, in September 2020,
  Ethereum miners made more money from transaction fees than from
  block rewards~\cite{bigfees}.}\fnsep\footnote{The inflation rate
  will become less predictable, however.}  In any case, because burned
fees are effectively a lump sum refund to ETH holders, the value of
ETH would be tied directly to the intensity of network usage.
Additionally, burned fees must be paid on-chain and in ETH, thereby
imbuing ETH with unique functionality.\footnote{In contrast, mere
  transfers between users and miners can be moved off-chain and paid
  using a different asset (e.g., USDT).}

Third, EIP-1559's base fee can serve as a difficult-to-manipulate
proxy for the current market-clearing gas price, which can in
turn enable a variety of new smart contracts (e.g., gas futures
markets).

Finally, 
there are well-documented incentive issues when transaction fees
dominate block rewards, for example the incentive for a miner to launch an
undercutting attack that forks a block with an unusually large amount
of transaction fees~\cite{ccs16}.  
By directing transaction fees away from
miners and to the network, EIP-1559 decreases the importance of
transaction fees to miners and makes such attacks less
attractive.\footnote{Though if the bulk of transaction fees come from
  a small number of transactions willing to pay much more than the
  base fee (e.g., submitted by front-running bots), such attacks will
  remain an issue.}

\subsection{The Escalator: EIP-2593}\label{ss:escalator}

EIP-2593 (a.k.a.\ the ``escalator'') is another proposal, orthogonal
to EIP-1559, that strives to improve the user experience through more
convenient fee estimation~\cite{2593spec}.\footnote{The idea behind
  this proposal was inspired by Miller and Drexler~\cite{MD88}.}
Its goal is not to change Ethereum's transaction fee mechanism (which
would remain a first-price auction), but rather to make bidding easier
for Ethereum users through a richer menu of bidding options.
Specifically, rather than a single gas price, an Ethereum
transaction would now come equipped with four bidding-related
parameters: 
\begin{itemize}

\item [(i)] the smallest block height at which the transaction is
valid, and a bid for that block; 

\item [(ii)] the largest block height at
which the transaction is valid, and a bid for that block.  

\end{itemize}
Bids for
all intermediate blocks are then determined automatically via linear
interpolation.  For example, a bid of~100 for block~10 and~150 for
block~20 induces the bids 105, 110, \ldots, 145 for blocks~11--19.
One would expect an impatient user to specify a relatively short
interval of blocks and a relatively high bid for the first block.  A
patient user, who favors a cheap price over immediate inclusion, would
presumably opt for a long interval and a low initial bid.

An Ethereum user could simulate the functionality of EIP-2593 by
rebroadcasting a transaction with successively higher gas prices.  The
goal of EIP-2593 is to automate this process in-protocol, eliminating
the added computational burden of resubmitted transactions.

EIP-2593 increases the number of bidding parameters relative to the
status quo---in effect, adding a rate of increase parameter to the
existing gas price parameter.  More parameters means
more in-protocol bidding options for users, but they also potentially
complicate the task of choosing a bidding strategy.\footnote{Hasu and
  Konstantopoulos~\cite{2593analysis} point out that another possible
  drawback of supporting richer bidding strategies is a decrease in
  privacy, with more clues about a user's preferences publicly
  available on-chain.}
EIP-2593 also locks users into a single type of bidding strategy (with
a linear bid increase), even though a user might be better served by
a different type of strategy (e.g., a more general concave or convex
function).

\begin{remark}[Combining EIP-1559 and EIP-2593]
EIP-2593 was initially proposed in part as an alternative to EIP-1559,
as a way to make fee estimation easier for users without introducing
any major changes to the Ethereum protocol.  The two proposals are
easily combined, however, by plugging in EIP-2593's linear bidding
strategies to set transaction tips in EIP-1559; the base fee would
evolve independently, according to the usual update rule
in~\eqref{eq:update}.
Because EIP-1559's tips are necessary primarily
in blocks with an excessively low base fee (see
Definition~\ref{def:low} and Theorem~\ref{t:1559uic}), EIP-2593's additional functionality may be
relevant only in the occasional period of rapidly increasing demand.\footnote{Monnot~\cite{2593sim} explores this
  design through simulations, with inconclusive results.}
\end{remark}

The scope of EIP-2593 is narrower than that of EIP-1559 and it makes much
less radical changes to the status quo.  The good news is that
the former proposal accordingly carries less risk than the latter; the
bad news is that it offers none of the side benefits listed in
Section~\ref{ss:benefits}.  For the specific objective of easier fee
estimation, the arguments currently justifying EIP-1559 
appear stronger and more rigorous than those
for EIP-2593.  In particular, because EIP-2593 retains the 
first-price transaction fee mechanism, there is no hope for ``obvious
optimal bids'' in the sense of the ``typically-UIC'' guarantee for the
1559 mechanism (Theorem~\ref{t:1559uic}).

\section{Conclusions}\label{s:conc}

Does EIP-1559 offer an improvement over Ethereum's current transaction
fee mechanism?  The biggest potential benefits of the proposed
changes are as advertised: easy fee estimation, in the form of an
``obvious optimal bid'' outside of periods of rapidly increasing
demand (Theorem~\ref{t:1559uic}); lower variance in transaction fees
due to increased flexibility in block size (Section~\ref{ss:toohigh});
game-theoretic robustness to protocol deviations and off-chain
agreements, both at the scale of a single block
(Theorems~\ref{t:1559mmic} and~\ref{t:1559ocaproof}) and of multiple
blocks (Section~\ref{s:collusion}); and reduced inflation due to fee
burning (Section~\ref{ss:benefits}).  

Most of the major risks in implementing EIP-1559 are the same as those
for any major change to the Ethereum protocol: implementation errors;
a fork caused by some parties rejecting the changes; extra complexity
at the consensus layer; additional parameters to be tweaked with every
network upgrade; and the spectre of unforeseeable downstream consequences.
Additional risks specific to EIP-1559 include the possibility of a
hostile reception by miners (due to lost revenue from burned
transaction fees) and a coordinated response
(Sections~\ref{ss:1559collusion}--\ref{ss:caveats}); and a new (if
expensive) attack vector enabled by variable-size blocks
(Sections~\ref{sss:chooserate}--\ref{sss:elastic}).

Reasonable people will disagree on whether the benefits of EIP-1559
justify the risks in adopting it.  Those who subscribe to a ``why fix
what isn't (too badly) broken'' philosophy may prefer to stick with
the status quo.  For those who believe that consensus-layer innovation
should continue to be a central part of Ethereum's future, however,
the arguments in favor of EIP-1559 are strong.

\end{document}